\documentclass[12pt]{article}
\usepackage{graphicx}
\usepackage[font=footnotesize]{caption}
\usepackage{natbib}
\usepackage{fancyhdr,amsmath,amssymb, amsthm,mathtools}
\usepackage{etoolbox}
\usepackage[colorlinks,citecolor=blue,urlcolor=blue]{hyperref}
\usepackage{booktabs}
\usepackage{tikz}
\usetikzlibrary{fit,positioning,calc}
\usepackage{scrextend}
\usepackage{hyperref}
\usepackage{chngcntr}
\usepackage{apptools}
\usepackage{bbm}
\usepackage{multirow}
\usepackage[T1]{fontenc}
\usepackage[export]{adjustbox}
\usepackage{mathrsfs}
\usepackage{mathrsfs,dsfont}
\usepackage{amsfonts}
\usepackage{stmaryrd}
\usepackage{mathrsfs,dsfont}
\usepackage{bbm,bm}
\usepackage{algorithm}
\usepackage{algpseudocode}
\usepackage{setspace}
\usepackage{subfigure}
\usepackage{upgreek}
\usepackage{textgreek}
\usepackage{makecell}
\usepackage{array}
\usepackage{enumitem}
\usepackage[export]{adjustbox}

\allowdisplaybreaks

\AtAppendix{\counterwithin{Lemma}{section}}

\newcommand{\E}{\mathds{E}}

\renewcommand{\P}{\mathds{P}}

\newcommand{\bx}{ {\bf x} }
\newcommand{\bix}{ {\boldsymbol x} }

\newcommand{\bP}{ {\bf P} }

\newcommand{\by}{ {\bf y} }
\newcommand{\biy}{ {\boldsymbol y} }
\newcommand{\bij}{ {\boldsymbol j} }
\newcommand{\bemme}{ {\bf m} }
\newcommand{\bim}{ {\boldsymbol m} }
\newcommand{\be}{{\bf e}}
\newcommand{\bie}{ {\boldsymbol e} }

\newcommand{\bY}{ {\bf Y} }

\newcommand{\bw}{ {\bf w} }
\newcommand{\biw}{ {\boldsymbol w} }

\newcommand{\br}{ {\bf r} }

\newcommand{\bk}{{\boldsymbol k}}

\newcommand{\bpi}{ {\boldsymbol \pi} }

\newcommand{\bLambda}{ {\boldsymbol \Lambda} }
\newcommand{\blambda}{ {\boldsymbol \lambda} }
\newcommand{\bPsi}{ {\boldsymbol \Psi} }

\newcommand{\bbeta}{ {\boldsymbol \beta} }

\newcommand{\bsigma}{ {\boldsymbol \sigma} }

\newcommand{\balpha}{ {\boldsymbol \alpha} }
\newcommand{\bxi}{ {\boldsymbol \xi} }
\newcommand{\bXi}{ {\boldsymbol \Xi} }
\newcommand{\bz}{ {\bf z} }
\newcommand{\biz}{ {\boldsymbol z} }

\newcommand{\brho}{ {\boldsymbol \rho} }

\newcommand{\bpsi}{ {\boldsymbol \psi} }

\newcommand{\bchi}{ {\boldsymbol \chi} }

\newcommand{\ddr}{\mathrm{d}}

\newcommand{\zdr}{\mathrm{z}}
\newcommand{\Ydr}{\mathrm{Y}}

\newcommand{\tdr}{\mathrm{t}}
\newcommand{\cdr}{\mathrm{c}}

\newcommand{\argmin}{\mbox{argmin}}

\newcommand{\appropto}{\mathrel{\vcenter{
  \offinterlineskip\halign{\hfil$##$\cr
    \propto\cr\noalign{\kern2pt}\sim\cr\noalign{\kern-2pt}}}}}

\definecolor{green}{rgb}{0.3, 0.73, 0.09}

\newcommand{\magenta}[1]{\!}

\newcommand{\prob}[1]{\P\left(#1\right)}
\newcommand{\cprob}[2]{\P\left(#1\left|#2\right.\right)}
\newcommand{\mean}[1]{\E\left[\,#1\right]}
\newcommand{\cmean}[2]{\E\left[\,#1\left|#2\right.\right]}

\DeclareFontFamily{OT1}{pzc}{}
\DeclareFontShape{OT1}{pzc}{m}{it}{<-> s * [1.10] pzcmi7t}{}
\DeclareMathAlphabet{\mathpzc}{OT1}{pzc}{m}{it}

\newcommand{\qt}[1]{`#1'}

\newtheorem{thm}{\textsc{Theorem}}[section]
\newtheorem{defi}{\textsc{Definition}}[section]

\newtheorem{prp}[thm]{\textsc{Proposition}}
\newtheorem{rmk}[thm]{\textsc{Remark}}
\newtheorem{exe}[thm]{\textsc{Example}}

\makeatletter
\providecommand{\leftsquigarrow}{%
  \mathrel{\mathpalette\reflect@squig\relax}%
}
\newcommand{\reflect@squig}[2]{%
  \reflectbox{$\m@th#1\rightsquigarrow$}%
}
\makeatother

\makeatletter
\newcommand\subsubsubsection{\@startsection{paragraph}{4}{\z@}{-2.5ex\@plus -1ex \@minus -.25ex}{1.25ex \@plus .25ex}{\normalfont\normalsize\bfseries}}
\makeatother

\newcommand{\blind}{1}

\expandafter\def\expandafter\normalsize\expandafter{%
    \normalsize%
    \setlength\abovedisplayskip{4pt}%
    \setlength\belowdisplayskip{4pt}%
    \setlength\abovedisplayshortskip{4pt}%
    \setlength\belowdisplayshortskip{4pt}%
}

\begin{document}

\def\spacingset#1{\renewcommand{\baselinestretch}%
{#1}\small\normalsize} \spacingset{1}

\if1\blind
{
	\title{\vspace{-20pt}
	\LARGE{\bf Exchangeable random permutations with an application to Bayesian graph matching}}
\author{Francesco Gaffi$^\star$, Nathaniel Josephs$^\dagger$, Lizhen Lin$^\ddagger$
  \hspace{.2cm}\\
    $^\star$Department of Economics, University of Bergamo, Italy\\
    $^\dagger$Department of Statistics, North Carolina State University, Raleigh, USA\\
    $^\ddagger$Department of Mathematics, University of Maryland, College Park, USA}
	\date{}
	\maketitle
} \fi

\if0\blind
{
	\bigskip
	\bigskip
	\bigskip
	\begin{center}
		{\LARGE\bf Exchangeable random permutations with an application to Bayesian graph matching}
	\end{center}
	\medskip
} \fi

\vspace{-25pt}

\begin{abstract}

We introduce a general Bayesian framework for graph matching grounded in a new theory of \emph{exchangeable random permutations}. Leveraging the cycle representation of permutations and the literature on exchangeable random partitions, we define, characterize, and study the structural and predictive properties of these probabilistic objects. A novel sequential metaphor, the \emph{position-aware generalized Chinese restaurant process}, provides a constructive foundation for this theory and supports practical algorithmic design. Exchangeable random permutations offer flexible priors for a wide range of inferential problems centered on permutations.
As an application, we develop a Bayesian model for graph matching that integrates a correlated stochastic block model with our novel class of priors. The cycle structure of the matching is linked to latent node partitions that explain connectivity patterns, an assumption consistent with the homogeneity requirement underlying the graph matching task itself. 
Posterior inference is performed through a node-wise blocked Gibbs sampler directly enabled by the proposed sequential construction. To summarize posterior uncertainty, we introduce \emph{perSALSO}, an adaptation of SALSO to the permutation domain that provides principled point estimation and interpretable posterior summaries. Together, these contributions establish a unified probabilistic framework for modeling, inference, and uncertainty quantification over permutations.
\end{abstract}
\noindent%
{\it Keywords:}  Bayesian nonparametrics, position-aware generalized Chinese restaurant process, exchangeable permutation probability function, correlated stochastic blockmodel. 

\spacingset{1.77} 
\vspace{-18pt}
\section{Introduction}\label{sec:intro}

The graph matching problem is a classic task that involves finding the correspondence between the vertices of two graphs, which has widespread applications in various fields, including computer vision, pattern recognition, bioinformatics, and social network analysis, to name a few.
To begin, consider two simple, undirected graphs $\mathcal{G}_1~=~(\mathcal{V},\mathcal{E}_1)$ and $\mathcal{G}_2=(\mathcal{V},\mathcal{E}_2)$ sharing the same set of nodes $\mathcal{V}$, with $\mathcal{E}_\ell\subset\mathcal{V}^2$ representing the edges of graph $\ell$. That is, for $\ell=1,2$,
    $\mathcal{E}_\ell=\{\{u,v\}\mid\exists \text{ an edge between $u$ and $v$ in graph } \ell\}$.
The \textit{exact} graph matching problem is to find, if it exists, an isomorphism between the graphs, \emph{i.e.} a bijection $\pi:\mathcal{V}\rightarrow\mathcal{V}$ such that
$\{u,v\}\in \mathcal{E}_1 \Leftrightarrow \{\pi(u),\pi(v)\}\in \mathcal{E}_2$.
The exact matching problem is NP hard, but the Weisfeiler-Lehman test is a well-known heuristic that has $O(|\mathcal V| \cdot |\mathcal E|)$ complexity \citep{leman1968reduction, huang2021short}.
More recently, subsequent quasi-polynomial algorithms for testing if two graphs are isomorphic have been discovered \citep{babai2016graph}.

A natural relaxation of the graph matching problem is to find a bijection $\pi$ of $\mathcal{V}$ in itself that minimizes the \emph{edge discrepancies} between $\mathcal{G}_1$ and $\mathcal{G}_2$.
Identifying the set of nodes $\mathcal{V}$ with the first $n$ integers $[n]:=\{1,\dots,n\}$, natural representations of the graphs are through their symmetric, binary adjacency matrices $\Ydr^{(1)}=(y^{(1)}_{uv})_{u,v=1}^n$ and $\Ydr^{(2)}=(y^{(2)}_{uv})_{u,v=1}^n$, defined as $y^{(\ell)}_{uv}=1$ if $\{u,v\}\in\mathcal{E}_\ell$, $y^{(\ell)}_{uv}=0$ if $\{u,v\}\notin \mathcal{E}_\ell$,
for $\ell=1,2$.
The \textit{inexact} graph matching problem is therefore to find $\hat\pi$, a \emph{permutation} of $[n]$ that is an element of the symmetric group $\mathcal{S}_n$, such that
\begin{equation}\label{eq:opt}
    \hat\pi=\displaystyle\argmin_{\pi\in \mathcal{S}_n}\sum_{1\leq u<v\leq n}y_{uv}^{(1)}(1-y_{\pi(u)\pi(v)}^{(2)})+(1-y_{uv}^{(1)})y_{\pi(u)\pi(v)}^{(2)}
\end{equation}
Since the objective in Equation~\eqref{eq:opt} is a version of the quadratic assignment problem, which is also NP hard, it is typically relaxed, for example by allowing for the minimization to occur over the set of all doubly  stochastic matrices rather than the set of permutation matrices.
An overview of various relaxations of Equation~\eqref{eq:opt} is given in \cite{qiao2021igraphmatch} and the literature on relaxed graph matching includes many different optimization techniques depending on the objective function \citep{lyzinski2014, lyzinski2015graph}.
Recently, \cite{ding2021efficient} proposed a novel probabilistic method that introduces the use of degree profiles, which are the empirical distributions of the degrees of a node's neighbors.
Under certain conditions, the authors prove that the smallest entries of the Wasserstein distance matrix between all pairs of degree profiles of the nodes in the two graphs recover the exact permutation with high probability.

One approach that is conspicuously missing from the literature is the Bayesian perspective.
There is a growing literature on Bayesian nonparametric methods for network data \citep{Dur(14),durante2017nonparametric, JMLR:v24:20-1206,josephs2023nested,josephs2023bayesian, amini2024hierarchical, durante2025partially}, but fewer references to Bayes-inspired techniques for graph matching.
\cite{pedersani2013} use Bayes' rule to match nodes based on their ``fingerprints," which are similar to degree profiles.
However, inference is cast as an optimization problem and no uncertainty quantification is provided.
\cite{lazaro2025probabilistic} introduce a probabilistic approach for aligning multiple networks to a latent blueprint network.
While the authors use MCMC to explore their posterior, a simple uniform prior is employed.
Moreover, the authors summarize their posterior by taking the element-wise mode, which itself is not guaranteed to be a permutation.
Neither of these methods provides a fully integrated Bayesian approach to graph matching.

To fill this gap, we introduce a new theory of exchangeable random permutations that, through their cycle representation, anchors our work to the prolific literature on exchangeable random partitions. In graph matching problems, exchangeability is a natural assumption: the
original node labels carry no intrinsic information, so any sensible model should treat nodes symmetrically, given the matching. In our
framework, this symmetry is expressed through the cycle representation;
we show that if a random permutation is exchangeable, then its cycle
structure must behave as an exchangeable partition of the nodes, and that the law of the cycles completely determines the
permutation distribution. Exchangeable laws on permutations are then
the least informative priors consistent with the desired symmetry, and at the
same time they admit a complete characterization via cycles and an
associated sequential predictive scheme. Sequential predictive constructions play a central role in modern
Bayesian modeling \citep{sandra}, because they encode exchangeability through a simple
one-step-ahead mechanism and naturally extend models beyond the observed
sample size. The predictive scheme we derive for exchangeable
permutations, the position-aware generalized Chinese restaurant process (PA-gCRP), parallels classical predictive constructions for
exchangeable random structures \citep{blackwell1973,pitman1995exchangeable}. We extend Kingman's notion of
partition structures \citep{Kin(78)} to \emph{permutation structures}, \emph{i.e.} sequences of Kolomogorov-consistent laws of random permutations of a growing number of objects. This combination of symmetry,
completeness, and sequential structure is what ultimately makes
posterior inference both interpretable and computationally feasible in
our framework.

This structured theoretical machinery allows us to address the graph matching problem by coupling the PA-gCRP with a correlated stochastic block model \citep{onaran2016optimal} that anchors the cycle structure of the matching to the latent block memberships. In other words, we conditionally constrain the search for the matching to the subspace of permutations that only match nodes in the same stochastic block, \emph{i.e.} exhibiting the same connectivity behavior. Such modeling choice follows naturally from the assumption of across-network homogeneity of the matched nodes, which is intrinsic to the graph matching task itself: one obviously needs to assume that the aligned nodes behave similarly in the two networks in order to reconstruct the network alignment from the observed connections.
Our Bayesian hierarchical construction entails the simultaneous learning of (i) a permutation that aligns the two networks, informed by (ii) a partition of the nodes, shared across networks coherently with the alignment. In this sense, our contribution can be framed in the context of multiplex network analysis \citep{kivela2014multilayer,macdonald22,pensky24,agterberg2025joint}, with the additional difficulty of an unknown alignment.
Together with a node-wise blocked Gibbs sampler, principled point estimation, and interpretable posterior summaries, our contributions establish a fully Bayesian framework for modeling, inference, and uncertainty quantification over permutations that paves the way towards applications beyond graph matching, such as record linkage \citep{fellegi69,steorts16,sadinle2017} and shuffled regression \citep{bala21,azadkia24,slawski24}. 

In summary, we provide:
(i) a full probabilistic characterization of exchangeable random
permutations through their cycle structure and an associated notion
of permutation structure, with a novel predictive sequential scheme (Section \ref{sec:erp}); (ii) a Bayesian model for graph matching that places a coherent prior
on the unknown permutation and yields full posterior uncertainty
quantification (Section \ref{sec:model}); (iii) a practical posterior sampling algorithm for permutations that
exploits the sequential representation to efficiently explore the space
of matching (Section \ref{sec:post}); (iv) a greedy posterior summary procedure for permutations (perSALSO), which
searches for a single representative matching informed by the full
posterior distribution (Section \ref{sec:summary}).

The proofs of our results, together with additional content, are provided in the supplementary material, whose sections are denoted by S.
Throughout, bold font denotes a random variable and $p(\cdot)$ denotes a probability mass function or a density function of a random element that is made clear by the argument, \emph{e.g.} $p(\Ydr)$ is the joint pmf of $\bY$. $\mathds P_\balpha$ is a conditional probability measure given $\balpha$, and $p_\alpha(\cdot)$ its probability mass or density function.

\vspace{-10pt}
\section{A class of nonparametric priors for permutations}\label{sec:erp}

In this section, we introduce a new class of nonparametric priors for permutations. We build upon the extensive literature on \emph{partition structures} \citep{Kin(78)}, which are sequences of consistent distributions of random partitions of a growing number of objects.
In Bayesian nonparametrics, the study of such stochastic processes, and of the random probability measures that induce them, has been leveraged for various tasks.
For example, they are used to directly model discrete data structures, as in species sampling models \citep{blackwell1973,Pit(96),lijoi07bio,fav09}, or for probabilistic clustering and density estimation by employing them in mixture models \citep{lo,escobar1995bayesian,Ish01,Lijoi05,Lijoi07,tommi} or in combination with stochastic block models \citep{nowicki2001estimation,Gen(19),legramanti2022extended, shen2025bayesian}. In this paper, we leverage them to model the distribution of the \emph{cycle structure} of a random permutation and as building blocks for defining \emph{exchangeable random permutations}.
Here, the exchangeability assumption, which is a cornerstone of Bayesian inference, translates into agnosticism to the original labeling of the objects that is being permuted.
For example, this is suitable for graph matching, where the unobserved permutation is a bijection mapping the nodes of one network into the nodes of the other: the original ordering of the nodes in both networks should not influence the inference of a matching.
In the following, we give some background on permutations, define exchangeable random permutations, and characterize them by defining suitable mathematical objects that identify their distribution. 

\subsection{The cycle structure of a permutation}

The set of all permutations of $[n]=\{1,\dots,n\}$, which we denote with $\mathcal{S}_n$, endowed with the composition operation $\cdot\,$, is known as the \emph{symmetric group}. Any $\pi\in\mathcal{S}_n$ is a bijection of $[n]$ in itself, and we adopt the function notation $\pi(i)=j$ to indicate that $\pi$ maps $i$ into $j$.  If $\sigma,\pi\in \mathcal{S}_n$ then $\sigma\cdot\pi(i)=\pi(\sigma(i))$. A permutation $\pi\in\mathcal{S}_n$ can be represented by a juxtaposition of disjoint \emph{cycles}.
Each cycle $(ijk\dots\ell)$ is a juxtaposition of elements of $[n]$, which we read as: $\pi(i)=j$, $\pi(j)=k$, \dots, $\pi(\ell)=i$. All elements of $[n]$ appear in this representation exactly once. For example, the permutation of the set $\{1,2,3,4\}$ given by
$
1\rightarrow 4,
2\rightarrow 1,
3\rightarrow 3,
4\rightarrow2
$
has cycle representation $(142)(3)$.
Since different cycle representations correspond to the same permutation, e.g. $(421)(3)$ and $(3)(142)$ give the same map as $(142)(3)$, one may add the following rules to ensure uniqueness:
(a) each cycle starts with its least element;
(b) cycles are ordered according to their first elements. The cycles of a permutation $\pi\in \mathcal{S}_n$ determine a \emph{partition} of $[n]$, which we will represent with the allocation vector $\zdr(\pi)=(z_1(\pi),\dots,z_n(\pi))$, where $z_i(\pi)$ is the ordinal of the cycle containing $i$.
Continuing with our example above,
$\zdr\left((142)(3)\right)=(1,1,2,1)$.
We call $\zdr(\pi)$ the \emph{cycle structure} of $\pi$. Notice that rule (b) implies that the labeling of the cycles in $\zdr(\pi)$ follows the so-called \emph{order of appearance}: cycle $j$ is the cycle that contains the least element that is not in cycle $j-1$. Such ordered allocation vectors represent partitions uniquely.
However, a partition $\zdr$ can correspond to $\prod_{j=1}^k(n_j-1)!$ different permutations, where $n_j$ is the cardinality of the $j$-th cycle of $\zdr$ and $k$ is the number of cycles in $\zdr$.
This is because the number of ways we can arrange the elements of the $j$-th cycle equals the number of \emph{cyclic} permutations of $[n_j]$, where a cyclic permutation is one whose cycle representation consists of a single cycle.

Next, we define a sequence of distributions for random permutations whose cycle structures are distributed as a partition structure, in the sense of \cite{Kin(78)}.

\vspace{-9pt}

\subsection{Exchangeable random permutations}

The properties of random permutations have been an object of great interest in the probabilistic literature: see, \emph{e.g.}, \cite{Arr(92)} for asymptotics of the cycle structure of a uniform random permutation, and \cite{Pit(19)} for asymptotics of more general random permutations. 

In this work, we are interested in defining a general class of distributions that can act as flexible priors for Bayesian inference on permutations, applicable to the graph matching problem, among other contexts where the labeling of the objects involved is not of interest, such as shuffled regression and record linkage. For this reason, we seek distributions on the space of permutations that: (i) incorporate the assumption of exchangeability of the statistical units involved (in our case, the node labels), and (ii) are consistent when the number of units grows.
The second desirable property is crucial because, in addition to unlocking theoretically sound predictions, it ensures a sequential construction that enables posterior computations (see Section \ref{sec:post}). 

Taking the lead from a construction described in \cite{Pit(96)} and borrowing the nomenclature of \cite{Kin(78)}, we wish to define general \emph{permutation structures} as infinite sequences of laws of random permutations $(\mathcal{L}_n)_{n\geq1}$ such that each $\mathcal{L}_n$ has support in $\mathcal{S}_n$ and it is invariant to relabeling, while the sequence fulfills a notion of \emph{consistency}, such as \qt{integrating out the $n+1$ object} from $\mathcal{L}_{n+1}$ one obtains exactly $\mathcal{L}_n$.
Each $\mathcal{L}_n$ will therefore be the law of an \emph{exchangeable random permutation}.
To this objective, we need to: (i) define a \emph{finitely} exchangeable random permutation, (ii) clarify in what sense a sequence of random permutations of a growing number of objects is \emph{consistent}, and (iii) finally define and characterize \emph{exchangeable} random permutations.

For task (i), we want to enforce distributional invariance with respect to label switching of the elements of the permutation. For example, if $\bpi$ is a random permutation in $\mathcal{S}_4$, then we want
$\P\left(\bpi=(132)(4)\right)=\P(\bpi=(1)(243))$, where we simply switched the element labels $1$ and $4$ (and used rules (a) and (b) for unique representation).

\begin{rmk}\label{rmk:conjugate}\normalfont
It is easy to see that label switching of the elements in a cycle representation is equivalent to the action of $\mathcal{S}_n$ on itself by conjugation: the permutation obtained switching the labels of the elements of the cycle representation of $\pi\in \mathcal{S}_n$ by $\sigma\in \mathcal{S}_n$ is $\sigma^{-1}\cdot\pi\cdot\sigma$. For example, we have $\pi'=\sigma^{-1}\cdot\pi\cdot\sigma$ with $\pi=(143)(2),\, \sigma=(1)(24)(3), \, \pi'=(123)(4)$.
\end{rmk}
\noindent In light of Remark \ref{rmk:conjugate}, the distribution of a finitely exchangeable random permutation of $[n]$ must be uniform on the conjugacy classes of $\mathcal{S}_n$. It is known that these are determined by the \emph{cycle types}.
For any permutation $\pi\in \mathcal{S}_n$, the cycle type of $\pi$ is the vector $\tdr(\pi)~=~(t_1(\pi),\dots,t_n(\pi))$ such that $t_i(\pi)$ is the number of cycles of length $i$ in $\pi$. In our example, $\pi=(143)(2)$ gives $t(\pi)=(1,0,1,0)$. 

\begin{rmk}\label{rmk:permind}\normalfont Notice that, if $\cdr(\pi)=(c_1(\pi),\dots,c_{k(\pi)}(\pi))$ is the the vector of cycle lengths of $\pi$ (where rule (i) fixes the ordering), and $k(\pi)$ is the number of its cycles, for any $\pi'\in\mathcal{S}_n$ we have $\tdr(\pi)=\tdr(\pi')$ if and only if $(c_1(\pi'),\dots,c_{k}(\pi'))=\left(c_{\sigma(1)} (\pi),\dots,c_{\sigma(k)}(\pi)\right)$ for some $\sigma\in\mathcal{S}_k$, where $k=k(\pi)=k(\pi')$. For example, $(143)(2)$ has (ordered) cycle lengths $(3,1)$ while $(1)(234)$ has $(1,3)$, but in fact they share the same cycle type $(1,0,1,0)$. 
\end{rmk}
\noindent Hence, we give the following definition.
\begin{defi}\label{def:rand_perm}
    A random permutation $\bpi\in \mathcal{S}_n$ is finitely exchangeable if
    $$\prob{\bpi=\pi}=\prob{\bpi=\pi'}$$ for any $\pi, \pi'\in \mathcal{S}_n$ such that $\tdr(\pi)=\tdr(\pi')$.
\end{defi}
\noindent Recalling the literature on partition structures  \citep[see \emph{e.g.}][]{Kin(78),Pit(96)}, we know that if $\bz=(\biz_1,\dots,\biz_n)$ is a finitely exchangeable random \emph{partition} of $[n]$ in the form of an allocation vector, then
$
    \prob{\bz=(z_1,\dots,z_n)}=\prob{\bz=\overline{\left(z_{\sigma(1)},\dots,z_{\sigma(n)}\right)}}
$
for any $\sigma\in \mathcal{S}_n$, where we indicate with $\overline{(z_{\sigma(1)},\dots,z_{\sigma(n)})}$ the only element of the orbit of $(z_{\sigma(1)},\dots,z_{\sigma(n)})$ with respect to the cluster label switching action whose labels are in order of appearance.
This leads to the following proposition.

\begin{prp}\label{prp:cy_str}
Let $\bpi$ be a random permutation of $[n]$. If $\bpi$ is finitely exchangeable, then its cycle structure $\zdr(\bpi)$ is an exchangeable random partition.
\end{prp}

For task (ii), we first need to define what it means to \qt{delete an object} from a permutation. In other words, we need to fix a projection of $\mathcal{S}_{n+1}$ onto $\mathcal{S}_{n}$ with respect to which the sequence of distributions $(\mathcal{L}_n)_{n\geq1}$ will be consistent. The most natural deletion procedure is the one obtained by canceling the object from the cycle representation of a permutation, which we formalize as follows.

\begin{defi}\label{def:del}
Given $\sigma\in \mathcal S_{n+1}$, we define $\mathpzc{d}(\sigma)\in \mathcal{S}_n$ as
\begin{equation}\label{eq:del}
\mathpzc{d}(\sigma)(i) = \begin{cases}
     \sigma(i) &\text{if}~i \neq \sigma^{-1}(n+1) \\
    \sigma(n+1) &\text{if}~i = \sigma^{-1}(n+1)
\end{cases}
\end{equation}
\end{defi}

\noindent For notational convenience, we also define the set $\mathcal{A}(\pi)$ of all permutations of $[n+1]$ that can be obtained from $\pi\in \mathcal{S}_n$ by adding $n+1$ to its cycle representation as follows.

\begin{defi}\label{def:add}
For any $\pi\in \mathcal{S}_n$, $\sigma\in\mathcal{A}(\pi)\subseteq \mathcal{S}_{n+1}$ if $\mathpzc{d}(\sigma)=\pi$.
\end{defi}

\noindent For example, if $\sigma=(143)(25)$, then $\mathpzc{d}(\sigma)=(143)(2)$.
On the other hand, if $\pi~=~(143)(2)$, then $\mathcal{A}(\pi)=\{(1543)(2),(1453)(2),(1435)(2),(143)(25),(143)(2)(5)\}$.
Notice that if $\pi~\in~\mathcal{S}_n$, then $|\mathcal{A}(\pi)| = n+1$. Hence, we have the following.
\begin{defi}
    A sequence $(\mathcal{L}_n)_{n\geq1}$ such that $\mathcal{L}_n$ is a distribution on $\mathcal{S}_n$ is consistent if, taking $\bpi_n\sim\mathcal{L}_n$ for any $n\geq1$, one has $\mathpzc{d}(\bpi_{n+1})\sim\mathcal{L}_n$ for any $n\geq1$. By extension we call $(\bpi_n)_{n\geq1}$ a consistent sequence of random permutations.
\end{defi}

\noindent Clearly, if $(\bpi_n)_{n\geq1}$ is a consistent sequence of random permutations, the support of $\bpi_{n+1}$ is a subset of $\mathcal{A}(\bpi_n)$ for any $n\geq1$.
\noindent Recalling that a consistent sequence of random \emph{partitions} $(\bz^{(n)})_{n\geq1}=((\biz^{(n)}_{1},\dots,\biz^{(n)}_{n}))_{n\geq1}$ is such that $(\biz^{(n+1)}_1,\dots,\biz^{(n+1)}_n)\overset{d}{=}\bz^{(n)}$ for any $n\geq1$, the following proposition can be stated.

\begin{prp}\label{prp:cons}
If $(\bpi_n)_{n\geq1}$ is a consistent sequence of random permutations, then $(z(\bpi_n))_{n\geq1}$ is a consistent sequence of random partitions.
\end{prp}

Finally, for task (iii), we give the following.

\begin{defi}\label{def:experm}
    A random permutation $\bpi\in\mathcal{S}_n$ is exchangeable if there exist a consistent sequence of random permutations $(\bpi^\star_m)_{m\geq 1}$ such that $\bpi^\star_m\in\mathcal{S}_m$ is finitely exchangeable for any $m\geq1$ and $\bpi^\star_n\overset{d}{=}\bpi$.
\end{defi}

\noindent Similarly, an exchangeable random \emph{partition} $\bz$ of $[n]$ can be defined as a the $n$-th dimensional projection of a consistent sequence of finitely exchangeable random partitions. Its distribution \citep[see \emph{e.g.}][]{pitman1995exchangeable,Pit(96)} is determined by an \emph{exchangeable partition probability function} (EPPF).
That is, for any $1\leq k\leq n$, there is a symmetric function $\varphi_k^{(n)}(n_1,\dots,n_k)$ with $\sum_{j=1}^k n_j=n$ that gives the probability of a realization of $\bz$ with $k$ clusters of sizes $(n_1,\dots,n_k)$. When $n$ varies, the collection $(\varphi^{(n)})_{n\geq1}$ verifies
\begin{equation}\label{eq:proj_eppf}
\varphi_k^{(n)}(n_1,\dots,n_k) = \Big(\sum_{j=1}^k \varphi_k^{(n+1)}(n_1,\dots,n_j+1,\dots,n_k) \Big) + \varphi_{k+1}^{(n+1)}(n_1,\dots,n_k,1)
\end{equation}
With the term EPPF we refer to the entire collection $(\varphi^{(n)})_{n\geq1}$.
Now we can give the following characterization of exchangeable random permutations.

\begin{thm}\label{thm:char}
    A random permutation $\bpi\in\mathcal{S}_n$ is exchangeable if and only if its distribution is determined by a collection of functions of the form
    \begin{equation}\label{eq:pmf}
        p(\pi)=p^{(n)}_k(n_1,\dots,n_k)=\frac{1}{\prod_{j=1}^k(n_j-1)!}\varphi^{(n)}_k(n_1,\dots,n_k)
    \end{equation}
    for any $k\in[n]$, where $(\varphi^{(m)})_{m\geq1}$ is an EPPF and $(n_1,\dots,n_k)$ are the cycle lengths of $\pi$.
\end{thm}

\begin{rmk}\normalfont\label{rmk:finex}
    The statement in Theorem \ref{thm:char} implies that to sample from an exchangeable random permutation, one can simply sample a cycle structure from an exchangeable random partition and then sample uniformly among all the partitions sharing that cycle structure. Notice, moreover, that such decomposition is necessary to maintain even \emph{finite} exchangeability of the random permutation (see Proof of Theorem \ref{thm:char} in Section \ref{sec:proofs}).
\end{rmk}

\noindent Given this characterization, we call a consistent sequence of distributions $(\mathcal{L}_n)_{n\geq1}$ of exchangeable random permutations a \emph{permutation structure} and the sequence $(p^{(n)})_{n\geq1}$ that determines it as in Theorem \ref{thm:char} an \emph{exchangeable permutation probability function} (EPerPF). Incidentally, we know that a sequence of random permutations whose distribution is determined by $(\mathcal{L}_n)_{n\geq1}$ does exist (see Proof of Theorem \ref{thm:char} in Section \ref{sec:proofs}). Finally, we give the following characterization of EPerPFs.
\begin{thm}\label{thm:eperpf}
    A collection $(p^{(n)}_k)^{n\geq1}_{k\leq n}$ where $p^{(n)}_k$ is a non-negative, symmetric function on $\{(n_1,\dots,n_k):\,\sum_{j=1}^nn_j=n\}$, is an EPerPF if and only if it satisfies
    \begin{equation}\label{eq:cons}
    p^{(1)}(1)=1,\qquad p^{(n)}_{k(\pi)}(\cdr(\pi))=\sum_{\sigma\in\mathcal{A}(\pi)}p_{k(\sigma)}^{(n+1)}(\cdr(\sigma))\quad\mbox{for any}\;\pi\in\mathcal{S}_n
    \end{equation}
\end{thm}

Permutation structures can be determined sequentially through their predictive probability functions. In the next section, we will illuminate this predictive scheme by modifying the restaurant metaphor used to represent the sequential construction of exchangeable random partitions.

\subsection{Predictive structure: the position-aware gCRP}\label{sec:predictive}

The Chinese restaurant process (CRP) \citep{blackwell1973,aldous1985} is a well-known generative scheme for a sequence of random partitions. In the metaphor, customers are seated at tables in a sequential fashion: as each new customer arrives, they either sit at an occupied table, or create a new table and sit alone. The former occurs with probability proportional to the number of customers already seated at the old table, while the latter occurs with probability proportional to a concentration parameter $\theta$.
At each discrete time point, one clusters together the $n$ seated customers according to the table they are sitting at, obtaining a realization of the random partition. This construction is directly related to the Pòlya urn scheme presented in \cite{blackwell1973} and, for every $n$, is equal in distribution to the random partition obtained from the ties on $n$ exchangeable samples directed by a Dirichlet process (see Example \ref{ex:dirichlet} and Section \ref{sec:ssp}).

When this mechanism is extended to describe the predictive scheme of any exchangeable random partition, it is called \emph{generalized Chinese restaurant process} (gCRP). If the random partition is Gibbs-type \citep{gnedinpitman,deblasi2015}, then the characteristic product form of the EPPF guarantees that the probability of sitting at a specific existing table is a function only of the number of customers at that table.
In general, whenever one has a partition structure, the urn scheme can be obtained as ratios of the consecutive probability mass functions. 

In \cite{Pit(96)}, the CRP is rephrased to generate a uniform random permutation. Here, we extend this strategy and modify the gCRP by considering a \emph{position-aware} version in which each new customer can either (i) sit \emph{to the left} of any of the already seated customers, and each seat at a given table is equiprobable, or (ii) create a new table and sit alone. Specifically, if at time $n$ there are $n_j$ customers at table $j$ and $k$ tables, the probability that the $(n+1)$-th customer sits at each of the available seats of table $j$ is 
\begin{equation}\label{eq:urn1}
    \frac{\varphi_k^{(n+1)}(n_1,\dots,n_j+1,\dots,n_k)}{{n_j}\,\varphi_k^{(n)}(n_1,\dots,n_k)}
\end{equation}
while the probability that they sit alone at a new table is
\begin{equation}\label{eq:urn2}
    \frac{\varphi_{k+1}^{(n+1)}(n_1,\dots,n_k,1)}{\varphi_k^{(n)}(n_1,\dots,n_k)}
\end{equation}
where $(\varphi^{(n)})_{n\geq 1}$ is an EPPF.
At each time point, we interpret the tables as cycles of a permutation, taking the specific seats at the table into account: if customer $i$ sits to the left of customer $j$ and to the right of customer $k$, then in that realization of the random permutation $j$ is mapped to $i$ and $i$ is mapped to $k$.
Clearly, the partition structure given by the division in tables is the one generated by a gCRP with EPPF $(\varphi^{(n)})_{n\geq 1}$.
The following proposition identifies the distribution of a random permutation $\bpi\in \mathcal{S}_n$ obtained through the position-aware gCRP (PA-gCRP) by formulas \eqref{eq:urn1}-\eqref{eq:urn2} with the pmf of an exchangeable permutation expressed in \eqref{eq:pmf}.

\begin{prp}\label{prp:pred}
The distribution of a consistent sequence of exchangeable random permutations $(\bpi_n)_{n\geq1}$ is determined by the following predictive scheme.
For any $\pi\in \mathcal{S}_n$ with $\cdr(\pi)~=~(n_1,\dots,n_k)$, if $\sigma^{(j)}\in\mathcal{A}(\pi)\subset\mathcal{S}_{n+1}$ such that $z_{n+1}(\sigma^{(j)})=j$, then
\begin{equation}\label{eq:pred}
    \cprob{\bpi_{n+1}=\sigma^{(j)}}{\bpi_n=\pi} = 
    \begin{cases}
        \frac{1}{n_j}\frac{\varphi_k^{(n+1)}(n_1,\dots,n_j+1,\dots,n_k)}{\varphi_k^{(n)}(n_1,\dots,n_k)}&1\leq j\leq k \\
        \frac{\varphi_{k+1}^{(n+1)}(n_1,\dots,n_k,1)}{\varphi_k^{(n)}(n_1,\dots,n_k)}&j=k+1    
    \end{cases}
\end{equation}
where $(\varphi^{(n)})_{n\geq1}$ is the EPPF of $(\zdr(\bpi_n))_{n\geq1}$.
If $\sigma\notin\mathcal{A}(\pi)$, then $\cprob{\bpi_{n+1}=\sigma}{\bpi_n=\pi}=0$.    
\end{prp}

\noindent In the light of this last characterization, we will refer to an exchangeable random permutation driven by the EPPF $(\varphi^{(n)})_{n\geq1}$ with the notation $\bpi\sim\text{PA-gCRP}(\varphi^{(n)})$.

Specifying an EPPF, and hence an EPerPF, can be challenging. A viable route is to define an almost surely discrete random probability measure, consider an exchangeable sequence directed by it and induce a random partition through the ties in such sequence, with the simple rule: two objects are in the same cycle if and only if the correspondent realizations in the sequence coincide.
Then a random permutation is obtained considering the uniform distribution conditionally on the cycle structure.
If the employed random probability measure is distributed as a Dirichlet process \citep{Fer(73)}, then we obtain the random permutation detailed in the following example.

\begin{exe}[Dirichlet process]\label{ex:dirichlet}\normalfont
If we employ the EPPF describing the partition induced by the ties in an exchangeable sequence directed by a Dirichlet process, we have
\begin{equation*}
    p^{(n)}_k(n_1,\dots,n_k)=\frac{\Gamma(\theta)}{\Gamma(\theta+n)}\theta^k  
\end{equation*}
where $\sum_{j=1}^kn_j=n$, for some concentration parameter $\theta>0$.
The predictive scheme is
$$p(\pi_{n+1}\mid\pi_n)=\frac{1}{n+\theta}\mathbbm{1}_{\{j\leq k\}}+\frac{\theta}{n+\theta}\mathbbm{1}_{\{j=k+1\}}$$
where $k=k(\pi_n)$ and $j=z_{n+1}(\pi_{n+1})$.
If $\theta=1$, we obtain the uniform distribution on $\mathcal{S}_n$. For $\theta>1$, permutations with more cycles are more probable with the identity being the most probable permutation. For $\theta<1$, permutations with fewer cycles are favored with the $(n-1)!$ cyclic permutations being the most probable ones. Such properties can be leveraged if one can elicit enough prior information. Otherwise, one can place a hyperprior on $\theta$ (see Section \ref{sec:post}). Notice that, for this particular permutation structure, the probability of the realizations does not depend on the specific cycle structure $\{n_1,\dots,n_k\}$, but only on the number of cycles $k$. If one can elicit some specific information about the composition of the cycles, it is sensible to go beyond the Dirichlet process and adopt one of the priors proposed in Examples \ref{exe:nsp}, \ref{exe:pyp}, and \ref{exe:gnp}, which are included in our implementation.
\end{exe}

\vspace{-15pt}
\section{A correlated SBM with PA-gCRP priors}\label{sec:model}

Our generative model for graph matching assumes that the cycles of the permutation of interest, which match the nodes of the two observed networks, coincide with the latent partition that characterizes the probability of connections between the nodes.
This modeling choice is in accordance with the homogeneity assumption that underlines every statistical graph matching method: the inferred permutation should minimize the edge discrepancies between the two observed networks because we assume implicitly that the same pairs of nodes should behave similarly.
In other words, one expects the matched nodes to have the same connectivity patterns in the two networks.
We codify this by generating the observed networks through a correlated stochastic block model (cSBM), where the latent division in connectivity blocks and the cycle structure of the permutation coincide so that only nodes in the same block can be matched.
This guarantees the simultaneous learning of a connectivity-informed matching and a model-based node clustering.
We present our cSBM with position-aware gCRP priors in an augmented formulation featuring a latent parent network.
An alternative formulation that directly gives the conditional distribution of one observed network given the other, whose equivalence carries over from the correlated Erd\H{o}s-R\'enyi model \citep{nate}, is given in Section \ref{sec:conditional}.

\subsection{Latent parent network formulation}

The cSBM, which is a generalization of the correlated Erd\H{o}s-R\'enyi model, can be defined as a pair of noisy observations of a common latent parent network.
The observed (undirected) networks are represented by binary (symmetric) adjacency matrices $\bY^{(1)}~=~(\biy^{(1)}_{uv})_{u,v=1}^n$ and $\bY^{(2)}=(\biy^{(2)}_{uv})_{u,v=1}^n$.
Rows and columns of $\bY^{(2)}$ are permuted through a random permutation $\bpi$.
The common parent network is encoded in a binary (symmetric) adjacency matrix $\bY=\left(\biy_{uv}\right)_{u,v=1}^n$ and its entries are defined as independent Bernoulli trials conditioned on a latent block membership of the nodes and the probabilities of connections between blocks.
This is the same as the classic Bayesian SBM \citep{nowicki2001estimation}, except here the latent blocks are specified by the cycle structure $\zdr(\bpi)$, where $\bpi$ is the matching between the rows of $\bY^{(1)}$ and the rows of $\bY^{(2)}$. Block connection probabilities are given by the matrix $\bXi=(\bxi_{jh})_{j,h=1}^{k(\bpi)}$, where $k(\bpi)$ is the (possibly random) number of cycles in $\bpi$. 
For conjugacy, the entries of $\bXi$ are given an independent beta prior, while the prior on $\bpi$ is generated by the position-aware generalized Chinese restaurant process (PA-gCRP) that we defined in Section~\ref{sec:erp}, driven by the EPPF $(\varphi^{(n)}_k)_{k\in[n]}$.
The hierarchical model is therefore:
\begin{equation}\label{eq:model1}
\begin{aligned}
    (\biy_{uv}^{(1)},\biy_{\bpi(u)\bpi(v)}^{(2)})_{u<v}\mid \bY,\bpi&\overset{ind}{\sim}\begin{cases}
        \text{Bern}(1-\beta)&\text{if }\,\biy_{uv}=1\\
        \text{Bern}(\alpha)&\text{if }\,\biy_{uv}=0
    \end{cases}\\
    (\biy_{uv})_{u<v}\mid\bXi,\,\zdr(\bpi)&\overset{ind}{\sim}\text{Bern}(\bxi_{z_u(\bpi)z_v(\bpi)})\\
    \bXi\mid\zdr(\bpi)&\overset{iid}{\sim}\text{beta}(a_\xi,b_\xi)\\
    \bpi&\sim \text{PA-gCRP}(\varphi^{(n)})
    \end{aligned}
\end{equation}
for $\alpha,\beta\in(0,1)$ with $\alpha<1-\beta$ and $a_\xi,b_\xi>0$.
If inference on the block connection probabilities is not of interest, one can integrate them out to obtain
\begin{equation}\label{eq:betabin}
    p(\Ydr \mid \zdr(\pi))= \prod_{j\leq h}^{k(\pi)}\frac{\mbox{B}(a_\xi+m_{jh},b_\xi+\overline{m}_{jh})}{\mbox{B}(a_\xi,b_\xi)}
\end{equation}
where $m_{jh}$ and $\overline m_{jh}$ are the number of edges and non-edges in $\bY$ between cycles $j$ and $h$. The Bayesian formulation makes it possible to simultaneously infer the error rates, if we further complete the model \eqref{eq:model1} with priors on $\balpha$ and $\bbeta$.
We restrict the support of such priors to $(0,1/2)$, since any higher noise level would cause a bias in any matching procedure that minimizes the edge discrepancy \citep{luna}.
A natural choice is the truncated beta distribution, which still allows a beta-binomial marginalization of the kind in \eqref{eq:betabin}. Hence, if $\balpha\sim \mbox{tbeta}(1/2;a_0,b_0)$ and $\bbeta\sim \mbox{tbeta}(1/2;a_1,b_1)$, we have the following joint distribution characterizing the model
\begin{equation}\label{eq:joint1}
\begin{split}
p(\Ydr^{(1)},\Ydr^{(2)},\Ydr,\pi)= \prod_{q=0,1}\frac{\mbox{B}\left(1/2;a_q+\bar e^{(1)}_{q}(\mathfrak{id})+\bar e^{(2)}_{q}(\pi),b_q+e^{(1)}_q(\mathfrak{id})+e^{(2)}_{q}(\pi)\right)}{\mbox{B}(1/2;a_q,b_q)}\\
\prod_{j\leq h}^{k(\pi)}\frac{\mbox{B}(1;a_\xi+m_{jh},b_\xi+\overline{m}_{jh})}{\mbox{B}(1;a_\xi,b_\xi)}
    p_{k(\pi)}^{(n)}(\cdr(\pi))
\end{split}
\end{equation}
where $\mbox{B}(q;a,b)=\int_0^{q}x^{a-1}(1-x)^{b-1}\ddr x$ and
\begin{equation}\label{eq:exponents}    e_q^{(\ell)}(\sigma)=\sum_{u<u'} \mathbbm{1}_{\{q\}}(y_{uu'})\mathbbm{1}_{\{q\}}(y^{(\ell)}_{\sigma(u)\sigma(u')})\qquad \bar e_q^{(\ell)}(\sigma)=\sum_{u<u'} \mathbbm{1}_{\{q\}}(y_{uu'})\mathbbm{1}_{\{1-q\}}(y^{(\ell)}_{\sigma(u)\sigma(u')})\end{equation}
are the concordant and discordant edges ($q=1$) and non-edges ($q=0$) in $(\Ydr, \Ydr^{(\ell)}_\sigma)$ for $\ell=1,2$ and we choose $\sigma=\mathfrak{id}$ the identity and $\sigma=\pi$ respectively in \eqref{eq:joint1}; $(n_1,\dots,n_{k(\pi)})$ are the cycle lengths of $\pi$. Since we consider the undirected case with no self-loops, only the upper-triangular part of the adjacency matrices is modeled, while we can set $\biy_{vu}=\biy_{uv}$ for $1\leq u<v\leq n$ and $\biy_{vv}:=0$ for $1\leq v\leq n$. Same for $\bY^{(1)},\bY^{(2)}_{\bpi}$.

\vspace{-18pt}

\section{Posterior inference}\label{sec:post}

In order to obtain samples from the posterior distribution of the random permutation $\bpi$ given the observed networks $\bY^{(1)}$ and $\bY^{(2)}$ and infer the latent parent network $\bY$ under the model in \eqref{eq:joint1}, we target the posterior
\begin{equation}\label{eq:post}
p(\pi, \Ydr ,\alpha,\beta \mid \Ydr^{(1)}, \Ydr^{(2)})
\end{equation}
We devise a node-wise blocked Gibbs sampler that leverages the sequential construction of the PA-gCRP given in Section \ref{sec:erp}. This greatly simplifies the challenging task of exploring the large space of permutations with an MCMC sampler. Specifically, we modify Definition \ref{def:del} to obtain a concept of \qt{permutation $\pi$ without considering node $v$} which allows node-wise blocked moves: $\mathpzc{d}_v(\pi)$ is a permutation of $[n]\setminus \{v\}$ such that
\begin{equation}\label{eq:delv}
\mathpzc{d}_v(\pi)(u) = \begin{cases}
     \pi(u) &\text{if}~u \neq \pi^{-1}(v) \\
    \pi(v) &\text{if}~u = \pi^{-1}(v)
\end{cases}
\end{equation}
For example $\mathpzc{d}_3\left((143)(2)\right)=(14)(2)$. Similarly, following Definition \ref{def:add}, for any $\sigma$ permutation of $[n]\setminus v$, we consider
\begin{equation}\label{eq:addv}
    \mathcal{A}_v(\sigma):=\left\{\pi\in\mathcal S_n:\mathpzc{d}_v(\pi)=\sigma\right\}
\end{equation}
In particular, $\mathcal{A}_v(\mathpzc{d}_v(\pi))$ is the set of all permutations that can be obtained by deleting node $v$ from the cycle representation of $\pi$ and then re-adding it in one of the $n$ available spots, in line with the sequential procedure described in Section \ref{sec:predictive}. For example $\mathcal{A}_3(\mathpzc{d}_3((143)(2)))=\{(134)(2),(143)(2),(14)(23),(14)(2)(3)\}$.

The sampler iterates the following steps.
\begin{itemize}[noitemsep]
\item[1.] For any node $v$, if $(\bpi,\bY)\in\mathcal S_n\times\{0,1\}^{n\times n}$ is the current state:
\begin{itemize}[noitemsep]
\item[(i)] sample a new permutation $\bpi^*\in\mathcal A_v(\mathpzc{d}_v(\bpi))$, conditionally on all the connections and the error rates;
\item[(ii)] sample sequentially a new row $\by^*_{v\cdot}=(\biy^*_{v1},\dots,\biy^*_{vn})$ for the latent parent network conditionally on the observed networks $\bY^{(1)}$ and $\bY^{(2)}$, the rest of the latent edges $\bY_{-v}$, the permutation $\bpi^*$ and the error rates $\balpha,\bbeta$.
\end{itemize}
\item[2.] Sample $\balpha, \bbeta$ from their full conditionals.
\end{itemize}

For step 1(i), we sample from
\begin{equation}\label{eq:mh}
\begin{split}
p_{\alpha,\beta}(\pi^\star \mid \Ydr, \Ydr^{(1)}, \Ydr^{(2)}, \mathpzc{d}_v(\pi))
\propto \frac{p_{\alpha,\beta}(\Ydr^{(2)} \mid \Ydr, \pi^\star)}{p_{\alpha,\beta}(\Ydr^{(2)}_{-v} \mid \Ydr_{-v}, \mathpzc{d}_v(\pi))} \frac{p(\Ydr \mid \zdr(\pi^\star))}{p(\Ydr_{-v} \mid \zdr(\mathpzc{d}_v(\pi)))} p(\pi^\star \mid \mathpzc{d}_v(\pi))
    \end{split}
\end{equation}
The first factor in \eqref{eq:mh} is the likelihood ratio of edge discrepancy between $\Ydr$ and $\Ydr^{(2)}_{\pi^\star}$ when re-adding $v$ to $\mathpzc{d}_v(\pi)$ and reduces to
\begin{equation}\label{eq:fact1}
(1-\beta)^{e^{(2)}_{1}(\pi^\star)-e^{(2)}_{1}(\mathpzc{d}_v(\pi))}\beta^{\bar e^{(2)}_{1}(\pi^\star)-\bar e^{(2)}_{1}(\mathpzc{d}_v(\pi))}\alpha^{\bar e^{(2)}_{0}(\pi^\star)-\bar e^{(2)}_{0}(\mathpzc{d}_v(\pi))}(1-\alpha)^{e^{(2)}_{0}(\pi^\star)-e^{(2)}_{0}(\mathpzc{d}_v(\pi))} 
\end{equation}
For computational convenience, we notice that for any $\pi^\star\in\mathcal A(\mathpzc{d}_v(\pi))$, setting $u^\star:=\pi^\star(v)$ and $u^{\star-1}:=[\mathpzc{d}_v(\pi)]^{-1}(u^*)$, if $u^\star\neq v$ then
\begin{equation}\label{eq:defstar}
   \begin{cases}
    \pi^\star(u) = \mathpzc{d}_v(\pi)(u) &\text{for all}~u \neq v,u^{\star-1} \\
    \pi^\star(u^{\star-1})=v
\end{cases} 
\end{equation}
while if $u^*=v$ we are just adding $(v)$ to the cycle representation of $\mathpzc{d}_v(\pi)$. Notice that $u^*\in[n]$ determines $\pi^\star$,  in fact $|\mathcal{A}_v(\mathpzc{d}_v(\pi))|=n$. Hence, for $u^\star\neq v$,
\begin{equation}\label{eq:dd}
    e_1^{(2)}(\pi^\star) = e_1^{(2)}(\mathpzc{d}_v(\pi))+\sum_{u=1}^ny_{vu}y^{(2)}_{u^\star \pi^\star(u)}-\sum_{u\neq v}y_{u^{\star-1} u}y^{(2)}_{u^\star \pi^\star(u)}+\sum_{u\neq u^\star} y_{u^{\star-1} u}y^{(2)}_{v\pi^\star(u)} 
\end{equation}
while if $u^*=v$, simply $e_1^{(2)}(\pi^\star) = e_1^{(2)}(\mathpzc{d}_v(\pi))+\sum_{u=1}^ny_{vu}y^{(2)}_{u^\star \pi^\star(u)}$. Similar equalities hold for the rest of the exponents in \eqref{eq:exponents}. Once the common $e_1^{(2)}(\mathpzc{d}_v(\pi))$ is available, one can use \eqref{eq:dd} to obtain the exponent for any $\pi^\star\in\mathcal A_v(\mathpzc{d}_v(\pi))$. In turn, \eqref{eq:dd} can be used at the next node-block move $v'\in[n]$, to obtain the next common component $e_1^{(2)}(\mathpzc{d}_{v'}(\pi))$.

The second factor in \eqref{eq:mh} is the SBM likelihood ratio of the realization $\Ydr$ when re-adding $v$ to the cycle structure. By \eqref{eq:betabin}, it equals
\begin{equation}\label{eq:fact2}
   \prod_{h=1}^{k(\pi^\star)} \frac{\mbox{B}(a_\xi+m^{-v}_{\zdr_v(\pi^\star)h}+r^\star_{vh},b_\xi+\overline{m}^{-v}_{\zdr_v(\pi^\star)h}+\overline{r}^\star_{vh})}{\mbox{B}(a_\xi+m^{-v}_{\zdr_v(\pi^\star)h},b_\xi+\overline{m}^{-v}_{\zdr_v(\pi^\star)h})}
\end{equation}
where $r^\star_{vh}$ and $\overline{r}^\star_{vh}$ are the total number of $\Ydr$-edges and $\Ydr$-non-edges, respectively, between node $v$ and all the nodes in the $h$-th cycle of $\pi^\star$.

The third factor is given by the predictive scheme of an exchangeable random permutation defined in Proposition \ref{prp:pred}
\begin{equation}\label{eq:fact3}
p(\pi^\star\mid\mathpzc{d}_v(\pi))=\begin{cases}
        \frac{1}{n_j^{-v}}\frac{\varphi_{k^{-v}}^{(n-1)}(n^{-v}_1,\dots,n^{-v}_j+1,\dots,n^{-v}_{k^{-v}})}{\varphi_{k^{-v}}^{(n-1)}(n_1^{-v},\dots,n^{-v}_{k^{-v}})}&1\leq j\leq k^{-v} \\
    
\frac{\varphi_{k^{-v}+1}^{(n)}(n^{-v}_1,\dots,n^{-v}_{k^{-v}},1)}{\varphi_{k^{-v}}^{(n-1)}(n^{-v}_1,\dots,n^{-v}_{k^{-v}})}&j=k^{-v}+1    
    \end{cases}
\end{equation}
where $j$ is the cycle of $\pi^\star$ where $v$ is added and $(n^{-v}_1,\dots,n^{-v}_{k^{-v}})$ are the cycle lengths of $\mathpzc{d}_v(\pi)$.

To initialize $\bpi$ we propose to sample uniformly among the permutations with cycle structure given by the partition obtained fitting to $\bY^{(1)}$ an unsupervised extended stochastic block model \citep[ESBM,][]{legramanti2022extended} with a Gibbs-type prior induced by $(\varphi^{(n)}_k)_{k=1}^n$.

For step 1(ii), \emph{i.e.} the sequential sampling of new latent connections of node $v$, we have
\vspace{-12pt}
\begin{equation}\label{eq:probedge}\mathds{P}_{\bpi,\balpha,\bbeta}\left[\biy_{vu}=1\left|\biy_{v1},\dots,\biy_{vu-1},\bY^{-v},\bY^{(1)},\bY^{(2)}\right.\right]=\frac{\mathfrak{p_1}}{\mathfrak{p_0}+\mathfrak{p_1}}
\end{equation}
where
\vspace{-12pt}
\begin{equation}\label{eq:prob1}
    \mathfrak{p_1}=(1-\bbeta)^{\biy^{(1)}_{vu}+\biy^{(2)}_{\bpi(v)\bpi(u)}}\bbeta^{(1-\biy^{(1)}_{vu})+(1-\biy^{(2)}_{\bpi(v)\bpi(u)})}\mbox{B}(1;a_\xi+\bim^{-vu}_{\zdr_v(\bpi)\zdr_u(\bpi)}+1,b_\xi+\overline{\bim}^{-vu}_{\zdr_v(\bpi)\zdr_u(\bpi)})
\end{equation}
\vspace{-12pt}
and
\begin{equation}\label{eq:prob0}  \mathfrak{p_0}=\balpha^{\biy^{(1)}_{vu}+\biy^{(2)}_{\bpi(v)\bpi(u)}}(1-\balpha)^{(1-\biy^{(1)}_{vu})+(1-\biy^{(2)}_{\bpi(v)\bpi(u)})}\mbox{B}(1;a_\xi+\bim^{-vu}_{\zdr_v(\bpi)\zdr_u(\bpi)},b_\xi+\overline{\bim}^{-vu}_{\zdr_v(\bpi)\zdr_u(\bpi)}+1)
\end{equation}
with $\bim_{jh}^{-vu}$ and $\overline{\bim}_{jh}^{-vu}$ denoting the number of $\bY$-edges and $\bY$-non-edges, respectively, between cycles $j$ and $h$, ignoring the interactions between node $v$ and nodes $u' \geq u$.

For step 2, the full conditional distributions of $\balpha,\bbeta$ are given by
\begin{equation}\label{eq:alphabeta}
\begin{aligned}
\balpha\mid\bpi,\bY,\bY^{(1)},\bY^{(2)}&\sim\mbox{tbeta}(1/2;a_0+\overline \bie^{(1)}_{0}(\mathfrak{id})+\overline \bie^{(2)}_{0}(\bpi),b_0+\bie^{(1)}_0(\mathfrak{id})+\bie^{(2)}_{0}(\bpi))\\
\bbeta\mid\bpi,\bY,\bY^{(1)},\bY^{(2)}&\sim\mbox{tbeta}(1/2;a_1+\overline \bie^{(1)}_{1}(\mathfrak{id})+\overline \bie^{(2)}_{1}(\bpi),b_1+\bie^{(1)}_1(\mathfrak{id})+\bie^{(2)}_{1}(\bpi))
\end{aligned}
\end{equation}

Finally, the predictive probabilities in \eqref{eq:fact3} depend on hyperparameters (\emph{e.g.} the concentration parameter $\theta>0$ in the Dirichlet process case, detailed in Example \ref{ex:dirichlet}) that may be given hyperpriors. For the simulations in Section \ref{sec:sim}, we employ a position-aware Dirichlet process prior with a gamma-distributed concentration parameter and we handle the conditional sampling of $\theta$ with a classic variable augmentation \citep{escobar1995bayesian}. Metropolis-Hastings steps can be added for placing hyperpriors on the parameters of different exchangeable random permutation distributions (see Examples \ref{exe:nsp}, \ref{exe:pyp}, \ref{exe:gnp}).

\begin{algorithm}[t]
	\caption{Node-wise blocked Gibbs sampler}\label{alg:1}
	\begin{algorithmic}
\footnotesize
\State Select EPerPF among Examples \ref{ex:dirichlet}, \ref{exe:nsp}, \ref{exe:pyp}, \ref{exe:gnp}
\State Initialization
\State\hspace{\algorithmicindent}$\bz$: partition from fitting an ESBM \citep{legramanti2022extended} on $\bY^{(1)}$
\State\hspace{\algorithmicindent}$\bpi$: random over permutations with cycle structure $\bz$
\State\hspace{\algorithmicindent}$\bY$: Bernoulli with probabilities $(\bY^{(1)}+\bY_{\bpi}^{(2)})/2$ \State\hspace{\algorithmicindent}$\balpha,\bbeta$: truncated beta priors
		\For{$s=1, \ldots, n_{\text{iter}}$ ($n_{\text{iter}}$: number of Gibbs iterations)}
		\For{$v\in\texttt{sample}(1 \ldots, n)$ ($n$: total number of nodes)} 
		\State remove node $v$: compute $\be^{(2)}(\mathpzc{d}_v(\bpi)),\overline \be^{(2)}(\mathpzc{d}_v(\bpi))$ from \eqref{eq:dd}, and $\bemme^{-v},\overline \bemme^{-v}$
        \For{$\bpi^\star\in\mathcal A(\mathpzc{d}_v(\bpi))$ in \eqref{eq:delv}-\eqref{eq:addv}} 
		\State 1. compute $\be^{(2)}(\bpi^\star),\overline \be^{(2)}(\bpi^\star)$ in \eqref{eq:dd} and the discrepancy likelihood ratio in \eqref{eq:fact1}
		\State 2. compute $\br_v, \overline{\br}_v$ and the sbm likelihood ratio in \eqref{eq:fact2}
		\State 3. compute the predictive probabilities in \eqref{eq:fact3} according to the selected EPerPF
        \EndFor
        \State sample new $\bpi$ from a multinomial with probabilities \eqref{eq:mh} 
		\For{$u=1,\dots,v-1$}
		\State 1. compute $\bemme^{-uv},\overline\bemme^{-uv}$
		\State 2. sample new $\by_{uv}$ from \eqref{eq:probedge}-\eqref{eq:prob1}-\eqref{eq:prob0}

		\EndFor
        \State compute the full $\bemme,\overline\bemme,\be^{(2)}(\bpi),\overline \be^{(2)}(\bpi)$
		\EndFor
		\State sample new $\balpha,\bbeta$ from \eqref{eq:alphabeta}
        \If{hyeperpriors employed} sample new hyperparameters (variable augmentation or MH step)
        \EndIf
		\EndFor
	\end{algorithmic}
\end{algorithm}

From this procedure, detailed in Algorithm \ref{alg:1}, one obtains samples $(\pi_s)_{s=1}^S$ and $(\Ydr_{(s)})_{s=1}^S$ from target posterior distribution \eqref{eq:post}. Explicit derivations of the full conditionals are given in Section \ref{sec:fullco}.

\vspace{-13pt}

\section{Posterior summary}\label{sec:summary}

Once a sample $(\pi_s)_{s=1}^S$ from the posterior distribution of $\bpi$ is obtained through Algorithm \ref{alg:1}, one is left with the non-trivial task of summarizing posterior information, determining a point estimate $\hat\pi$, and communicating its uncertainty. The calculation of Fr\'echet means in $\mathcal S_n$ becomes computationally unfeasible even for moderate values of $n$. In particular, when the permutations are interpreted as rankings (\emph{i.e.} the ordering of their image matters), this coincides with the problem of finding a consensus ranking \citep{pmlr}, and is related to the Kemeny ranking problem, which is known to be NP-complete \citep{complexkemeny}.
We propose \emph{perSALSO}, a greedy approach that follows the sequential construction of our model and is based on an extension of SALSO \citep{salso} to the space of permutations.

\subsection{perSALSO}\label{sec:persalso}
As in SALSO, we propose a greedy search based on the local minimization of the Monte Carlo estimate of the posterior expected distance, that is
\begin{equation}\label{eq:postexpdist}
    f_C(\sigma):=\frac{1}{S}\sum_{s=1}^Sd_C(\pi_s,\sigma)\approx\cmean{d_C(\bpi,\sigma)}{\bY^{(1)},\bY^{(2)}}
\end{equation}
In particular we employ the \emph{Cayley distance} $d_C(\cdot,\cdot)$, which is defined as the minimum number of transpositions, \emph{i.e.} $\rho\in\mathcal S_n$ with $t(\rho)=(n-1,1)$, needed to obtain one permutation from another by composition. Cayley distance is a convenient choice in our context because it is \emph{bi-invariant} (see Section \ref{sec:dist}) and can be computed easily, due to Cayley's identity
$d_C(\pi,\sigma)= n - k(\sigma^{-1}\cdot\pi)$,
which saves the minimization step in the definition. Nonetheless, its use is suboptimal for posterior samples obtained with our model, since similarity of the cycle structure does not generally imply lower Cayley distance. We discuss this in Section \ref{sec:dist}. 

Our perSALSO features sequential moves in line with our generative model and, as the original algorithm, consists of three phases.
\begin{itemize}[noitemsep]
    \item[(1)] \emph{Initialization}: an initial proposal is constructed sequentially, adding one node at a time, minimizing the posterior expected Cayley distance at each step. 
    \item[(2)] \emph{Sweetening}: each node is removed and re-added in the position that minimizes the posterior expected Cayley distance. This phase terminates when a full pass over the nodes does not lead to an improvement in terms of posterior expected Cayley distance. 
    \item[(3)] \emph{Zealous updates}: for a fixed number of times, an entire cycle is removed and its nodes are re-added sequentially, as in the initialization phase.
\end{itemize}
Both phases (1) and (3) require the concept of distance between permutations of two different sets of objects, in particular one being a subset of the other. Here, our approach matches the constructive mechanism described in Section \ref{sec:erp}: to quantify the distance between a permutation $\sigma$ of $\mathcal{C}$ and a permutation $\sigma'$ of $\mathcal{C'}$ with $\mathcal C\subset\mathcal{C'}\subseteq[n]$, we project $\sigma'$ on $\mathcal{C}$ using a composition of deletions of the kind in \eqref{eq:delv}, and then consider the Cayley distance between such projection and $\sigma$.

Specifically, for phase (1) we proceed as follows. 
If $\sigma\in\mathcal S_n$, let $\mathpzc{d}_\mathcal{C}(\sigma):=\mathpzc{d}_{v_1}\cdots \mathpzc d_{v_C}(\sigma)$ for any $\mathcal C=\{v_1,\dots,v_C\}\subset[n]$. For any
 random ordering of the nodes $\brho\in\mathcal{S}_n$, let $\mathcal V^{(\brho)}_i:=\{{\brho(1)},\dots,{\brho(i)}\}\subseteq[n]$ for any $i=1,\dots,n$ and $\mathcal V_0^{(\brho)}=\varnothing$. Then we find 
\begin{equation}\label{eq:init}
 \hat \pi_{i}^{(1)}=\argmin_{\displaystyle\sigma\in\mathcal{A}_{\brho(n-i+1)}(\hat\pi_{i-1}^{(1)})}\; \frac{1}{S}\sum_{s=1}^Sd_C(\mathpzc{d}_{\mathcal{V}^{(\brho)}_{n-i}}(\pi_s),\sigma)
\end{equation}
for $i=2,\dots,n$, while $\hat \pi^{(1)}_1=({\brho(n)})$. Our initialization is then $\hat\pi^{(init)}:=\hat\pi^{(1)}_n$. In words, we start from the trivial permutation of one node, and at every step we choose, among the permutations obtainable adding an unmatched node to the previous step permutation, the one that minimizes the average Cayley distance from a \emph{pruned} version of the sample only including the already-matched nodes and the newly added one. The (reverse) order in which the nodes are added is given by $\brho$.

In phase (2), we simply find 
\begin{equation}\label{eq:sweet}
     \hat \pi_{i}^{(2)}=\argmin_{\displaystyle\sigma\in\mathcal{A}_{\brho'(i)}(\mathpzc{d}_{\brho'(i)}(\hat\pi_{i-1}^{(2)}))}f_C(\sigma)
\end{equation}
for $i=1,\dots,n$ starting from $\hat\pi^{(2)}_0=\hat\pi^{(init)}$, for some $\brho'\in\mathcal S_n$. If $f_C(\hat\pi^{(2)}_n)>f_C(\hat\pi^{(2)}_0)$ we redefine $\hat\pi^{(2)}_0:=\hat\pi^{(2)}_n$ and re-iterate \eqref{eq:sweet} with a different $\brho''\in\mathcal S_n$. Otherwise, we stop and set $\hat\pi^{(sweet)}:=\hat\pi^{(2)}_n$.

Finally, in phase (3), if $(\mathcal{C}_1,\dots,\mathcal{C}_k)$ are the cycles of $\hat\pi^{(sweet)}$ with lengths $(n_1,\dots,n_k)$, we choose one randomly, say the $\bij$-th, and, similarly to \eqref{eq:init}, we find
\begin{equation}\label{eq:zeal}
 \hat \pi_{i}^{(3)}=\argmin_{\displaystyle\sigma\in\mathcal{A}_{\brho_\bij(n_\bij-i+1)}(\hat\pi_{i-1}^{(3)})} \; \frac{1}{S}\sum_{s=1}^Sd_C(\mathpzc{d}_{\mathcal{V}^{(\brho_\bij)}_{n_\bij-i}}(\pi_s),\sigma)
\end{equation}
for $i=2,\dots,n_\bij$, with $\hat\pi^{(3)}_1$ being the permutation obtained adding the cycle $(\brho_\bij(n_\bij))$ to $\mathpzc{d}_{\mathcal C_\bij}(\hat\pi^{(sweet)})$, where now $\brho_\bij$ is a random permutation of $\mathcal C_\bij$ and $\mathcal V^{(\brho_\bij)}_i:=\left\{[n]\setminus\mathcal C_\bij\right\}\cup \{{\brho_\bij(1)},\dots,{\brho_\bij(i)}\}\subseteq[n]$ for any $i=1,\dots,n_\bij$.
The minimization in Equation~\eqref{eq:zeal} is repeated, starting from $\hat\pi^{(3)}_{n_\bij}$, for a fixed number of times, with the last giving $\hat\pi^{(zeal)}$. The three phases can be repeated in parallel and the minimizer of $f_C(\cdot)$ among the obtained $\hat\pi^{(zeal)}$'s is our proposed point estimate $\hat\pi$.

The sequential nature of perSALSO allows for computationally feasible exact locally optimal steps, since the search at each step is reduced to at most $n$ evaluations (in phase (2)) of the average Cayley distance from the sample. Empirical validations for a similar greedy approach in the partition framework have been given, \emph{e.g.} in \cite{Ras(18)} and \cite{salso}.
To further cut down the complexity and leverage the specific nature of our posterior sample, in Section \ref{sec:speed} we provide a fast two-step approximation that first gives a point estimate of the cycle structure and then employs it to reduce the search at each step to $n_j$ evaluations even in phase (2). A detailed pseudocode is given in Section \ref{sec:speed}, together with speed-up techniques to cut the computational cost of each evaluation of $f_C(\cdot)$.

\vspace{-0pt}

\section{Simulations}\label{sec:sim}

In this section, we present the results of our model on simulated data.
We perform posterior inference following Algorithm \ref{alg:1}, specifying a position-aware Dirichlet process EPerPF (see Example \ref{ex:dirichlet}), with a gamma hyperprior on the concentration parameter.
For each simulation, we run our sampler for 10,000 iterations.
We compute the posterior summary described in Section~\ref{sec:summary} with a burn-in of 2,000 iterations and thinning every 10 samples.

We evaluate our performance using three metrics:
(i) Cayley distance between inferred and true permutation;
(ii) number of edge discrepancies between the aligned networks, \emph{i.e. Frobenius distance} between the aligned adjacency matrices $\Vert \bY^{(1)} - \bY^{(2)}_{\hat \bpi}\Vert_F$, which is a common objective function for the inexact graph matching problem;
and (iii) \textit{normalized mutual information} (NMI) between inferred and true communities, which measures partition agreement and ranges from 0 to 1.
In all figures, metrics (i) and (ii) are normalized for comparison across different networks.
For posterior uncertainty, one can consider the posterior expected values of these metrics. 

We compare our model to an indefinite relaxation approach \citep{qiao2021igraphmatch} and matching via degree profiles \citep{ding2021efficient}.
We refer to these as \textit{Relaxed} and \textit{Degree Profiles}, respectively.
An implementation of our model, \textit{Bayes}, along with the files to recreate our simulations, is available at the GitHub repository \texttt{francescogaffi/BayesGM}.

\vspace{-7pt}

\subsection{Mixing and uncertainty quantification}

We begin by analyzing a single run from our model to evaluate MCMC mixing, matching and parent network recovery, and uncertainty quantification.
We sample a parent network $\bY$ from an SBM with block probabilities given in Figure~\ref{fig:sim_sbm}, mimicking the communities found in the complex criminal network analyzed in  \cite{legramanti2022extended} and \cite{durante2025partially}, that include core-periphery, assortative, and disassortative block structures.
In particular, the parent network has $n = 40$ nodes with $k = 7$ communities and edge probabilities ranging from $0.1$ to $0.9$.
We then sample two noisy networks $\bY^{(1)}$ and $\bY^{(2)}$ following the cSBM framework with $\alpha = \beta = 0.01$.
Finally, we sample a random permutation with cycles given by the community labels.

\begin{figure}[t]
    \centering
    \includegraphics[width=\linewidth]{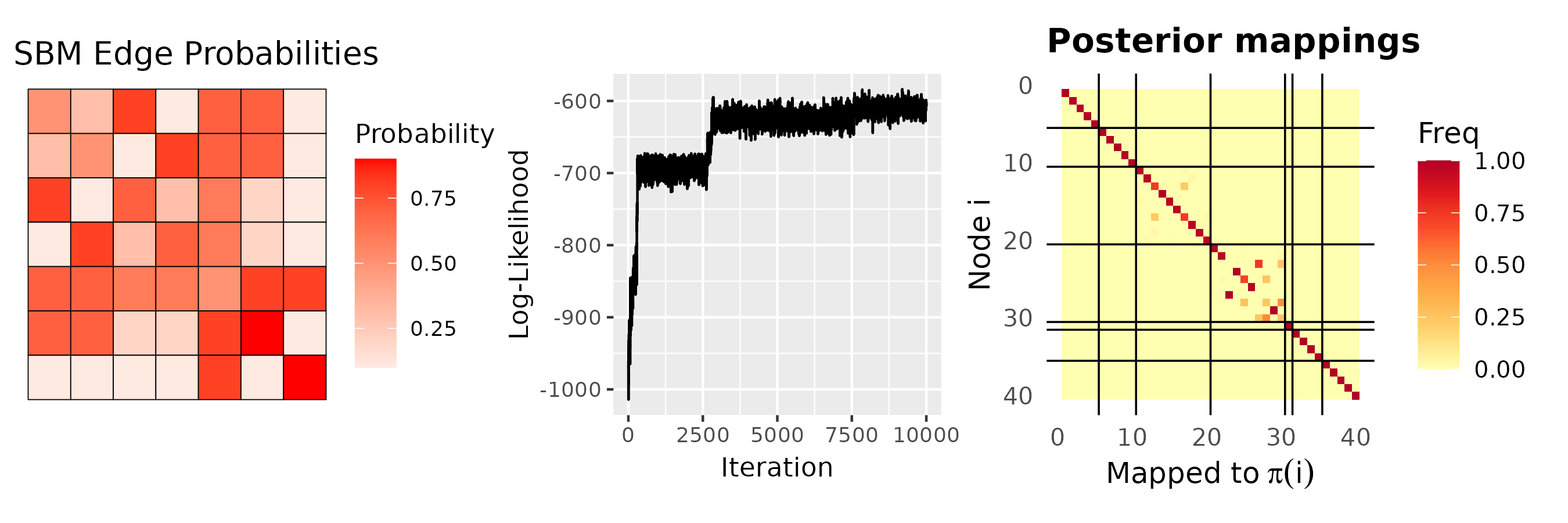}
    \caption{(Left) Block matrix of probabilities for underlying SBM. (Center) Trace plot of log-likelihood. (Right) Posterior mapping frequency, where the nodes on the x-axis are permuted through $\pi^*$.}
    \label{fig:sim_sbm}
\end{figure}

Figure~\ref{fig:sim_sbm} shows the trace plot of the log-likelihood over 10,000 iterations.
Our sampler quickly identifies a mode with much higher likelihood than the initialization.
After about 2,500 iterations, it jumps to another mode, and then once more around the 7,500th iteration.
Throughout, the chain mixes well. Table~\ref{tab:sbm_sim} collects the results for this simulation.
The Cayley distance is 2, meaning there are just two transpositions needed to transform our inferred permutation to the true permutation, and closely matches the posterior expected Cayley distance of 2.1.
Similarly, the average Frobenius norm for our inferred permutation (2.3) is close to the true permutation (0.7), \emph{i.e.} our alignment has a similar number of edge discrepancies.
As Figure~\ref{fig:sim_sbm} also shows, the true cycle structure is learned perfectly, implying an NMI of 1.
In other words, the only nodes that are assigned incorrectly are still assigned to the correct cycle.
Importantly, these errors correspond to where our model is most uncertain.
Finally, an AUC of 0.97 shows almost perfect reconstruction of latent parent network.

\begin{table}[h]
\centering
\begin{adjustbox}{width=1\textwidth}
\begin{tabular}{|c|c|c|c|c|c|}
\hline
 $d_C(\hat\pi,\pi^*)$ & $\cmean{d_C(\bpi,\pi^*)}{\bY^{(1)},\bY^{(2)}}$ &  $\|\bY^{(1)}-\bY^{(2)}_{\hat\pi}\|_F$
 &$\|\bY^{(1)}-\bY^{(2)}_{\pi*}\|_F$  & $\mbox{NMI}(\zdr(\hat\pi),\zdr(\pi^*))$ & \textbf{AUC} \\ \hline
\multicolumn{1}{|r|}{2} & \multicolumn{1}{r|}{2.1} & \multicolumn{1}{r|}{0.7} & \multicolumn{1}{r|}{2.3} & \multicolumn{1}{r|}{1} & \multicolumn{1}{r|}{0.97} \\ \hline
\end{tabular}
\end{adjustbox}
\caption{Cayley distance between the point estimate $\hat\pi$ and $\pi^*$, expected posterior Cayley distance from $\pi^*$, Frobenius distance between the observed matrices aligned with $\hat\pi$, to be compared to the oracle Frobenius distance with $\pi^*$ alignment, NMI between the cycle structure of the point estimate $\zdr(\hat\pi)$ and the cycle structure of $\pi^*$, the area under the ROC curve for edge reconstruction of the latent parent network.}
\label{tab:sbm_sim}
\end{table}

\vspace{-18pt}

\subsection{Number of nodes}

We next evaluate the performance of our method as the number of nodes increases.
We set the true permutation as $\pi^* = (2, \ 3, \ \cdots, \ n/2, \ 1)(n/2 + 2, \ \cdots, \ n, \ n/2 + 1)$.
Hence, the true cycle structure is $\zdr(\pi^*) = (1, \ldots, 1, 2, \ldots, 2)$.
We then sample $\bY$ from an SBM with communities given by $\zdr(\pi^*)$, where the probabilities of an edge within community and between communities are 0.6 and 0.1, respectively.
Finally, we sample $\bY^{(1)}$ and $\bY^{(2)}$ from a cSBM with $\alpha = \beta = 0.05$.

\begin{figure}[h]
    \centering
    \includegraphics[width=\linewidth]{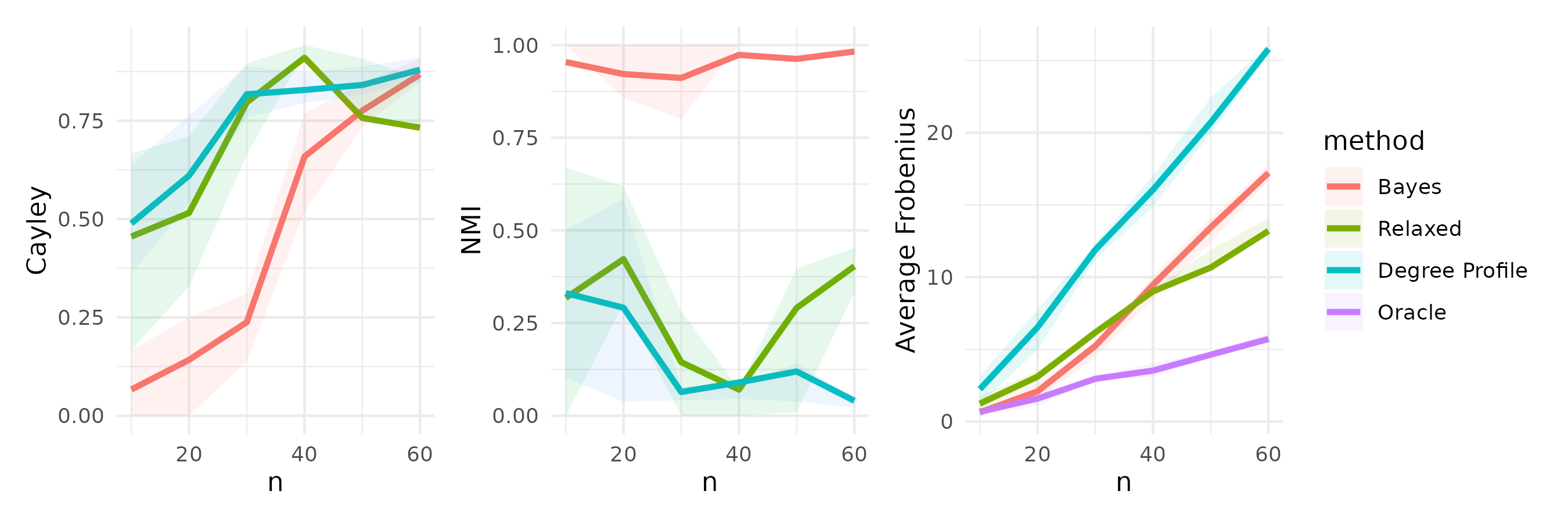}
    \caption{Simulation results as the number of nodes increases.}
    \label{fig:sim_scale}
\end{figure}

Figure~\ref{fig:sim_scale} shows the results over 10 replicates as the number of nodes increases from $n = 10$ to $60$.
First, we see that for small networks, \emph{Bayes} consistently recovers a permutation that is much closer, in Cayley distance, to the true permutation than the existing methods.
Second, we see that \emph{Bayes} finds a permutation with a similar number of edge discrepancies as \emph{Relaxed}, whose objective function is directly related to this minimization.
This is noteworthy because: (i) our model minimizes edge discrepancies without being explicitly designed to do so and (ii) the two inferred permutations 
have a similar number of edge discrepancies despite being very different.
The latter point can be seen by observing the difference in Cayley distance to the truth, but also the difference in NMI: our model gets the cycle structure almost exactly correct (NMI close to 1) while the other methods learn permutations with cycle structures that depart substantially from the true one.

\subsection{Noise level}

We repeat the previous simulation study except instead of increasing the number of nodes, we increase the noise level, fixing $n = 30$ and letting $\beta$ vary from $0.01$ to $0.4$.
We assume edge unbiasedness, \emph{i.e.} the expected number of observed edges is equal to the number of parent edges.
Hence, we set $\alpha =  \frac{\beta|\mathcal E|}{\binom{n}{2} - |\mathcal E|}$, where $\mathcal E$ is the set of parent edges.

\begin{figure}[h]
    \centering
    \includegraphics[width=\linewidth]{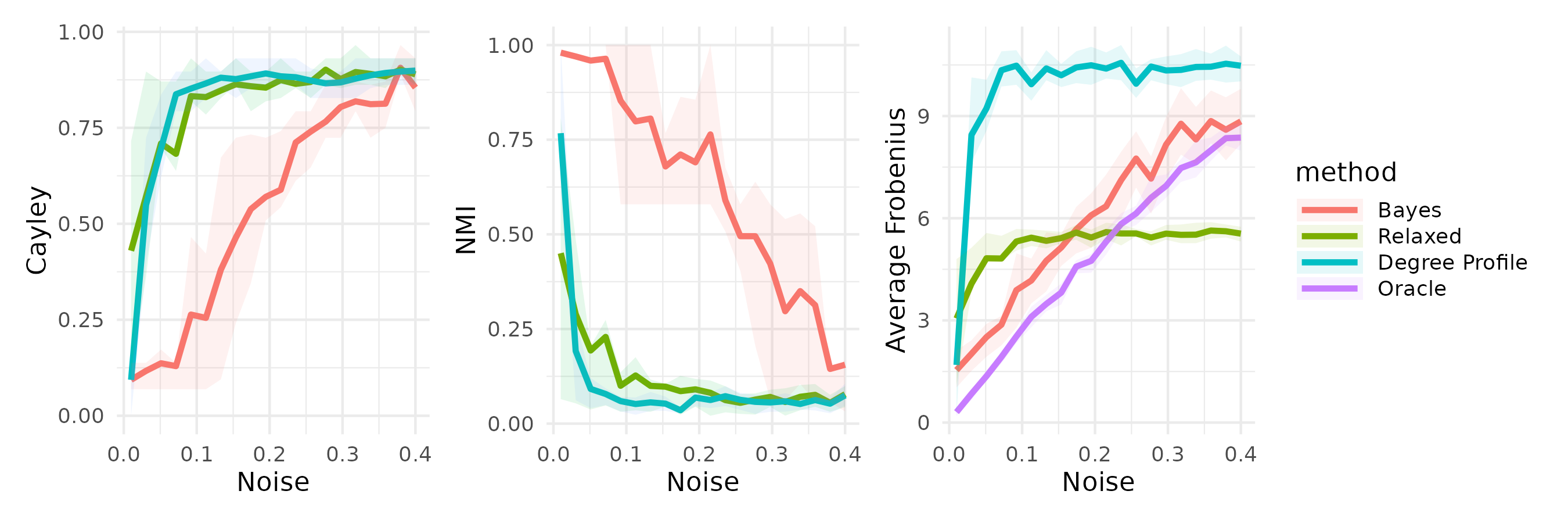}
    \caption{Simulation results as the cSBM noise level increases.}
    \label{fig:sim_corr}
\end{figure}

We plot the results in Figure \ref{fig:sim_corr} where each noise level have 28 replicates.
As expected, we see that all the methods perform worse as the noise level increases (left to right).
However, in terms of Cayley distance, our method decreases linearly with the noise while the other methods do exponentially worse.
We also see that \emph{Bayes} closely tracks the oracle rate in terms of Frobenius distance, whereas \emph{Relaxed} underfits at low noise and overfits at high noise levels.
As before, neither of the other methods infers a cycle structure that closely matches the true community structure.

\subsection{Concurrence of cycle and community structures}

Finally, we assess the robustness of our model to the assumption that the permutation cycles align with the underlying community structure.
To achieve this, we consider the same type of permutation $\pi^\star$ on $n = 21$ nodes as in the previous simulations except with 3 equally-sized cycles.
Hence, the true cycle structure is $\zdr(\pi^\star) = (1, \ldots, 1, 2, \ldots, 2, 3, \ldots, 3)$.
We then create a sequence of community structures that deviates from this cycle structure. In particular, we start with $\bz^{(0)} = \zdr(\pi^\star)$.
Then, given $\bsigma\in\mathcal S_n$, we define $\bz^{(u)}$ recursively as
$\biz^{(u)}_v =
        \biz^{(u-1)}_v$ for $v \neq \bsigma(u)$ and
        $\biz^{(u)}_{\bsigma(u)}=\bix^{(u)}
$
for $u\in[n]$, where $\bix^{(u)}$ is sampled uniformly from the unique values of $\bz^{(u-1)}_{-\bsigma(u)}$ plus a new community label.
In other words, at each step, we reallocate $\bsigma(u)$ to a different community, or to a new one.
This results in a sequence $\bz^{(0)}, \ldots, \bz^{(n)}$ that slowly deviates from the true cycle structure $\zdr(\pi^\star)$.
We sample a parent network with communities given by $\bz^{(u)}$ and two noisy children with $\alpha = \beta = 0.05$.
We repeat this 28 times for each $\bz^{(u)}$.

\begin{figure}[h]
    \centering
    \includegraphics[width=\linewidth]{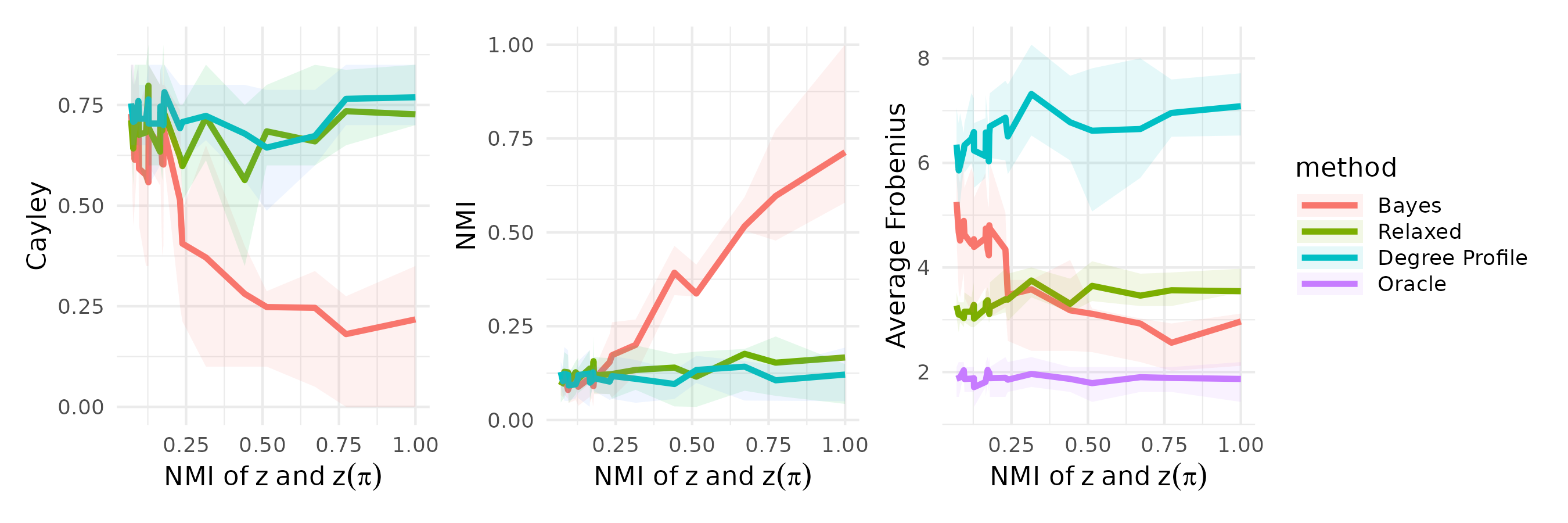}
    \caption{Simulation results as we increase the concurrence of the latent permutation cycle structure and the underlying community structure.}
    \label{fig:sim_robust}
\end{figure}

The results are given in Figure~\ref{fig:sim_robust}.
The x-axis is the NMI between $\bz^{(u)}$ and $\zdr(\pi^\star)$.
Hence, the model becomes more properly specified from left to right.
As we expect, neither \emph{Relaxed} nor \emph{Degree Profile} is sensitive to the overlap in cycle and community structue, while \emph{Bayes} improves in all respects as $\bz$ and $\zdr(\pi^\star)$ become more similar.
Importantly, \emph{Bayes} still outperforms the other methods as long as the NMI exceeds 0.25, which corresponds to moderate agreement.

\section{Conclusions and future research}
\label{sec:conc}

\vspace{-10pt}

This paper develops a unified Bayesian framework for inference on permutations, grounded in a new theory of exchangeable random permutations. We provide a complete probabilistic characterization, a constructive sequential representation, and practical algorithms for posterior inference and summarization. When applied to graph matching, this framework yields a fully Bayesian alternative to optimization-based approaches, delivering coherent uncertainty quantification and competitive alignment performance.

Several directions for future research naturally emerge from this work. First, given its predictive nature, the proposed framework can be further developed for graph matching problems with missing edges, partial alignment, or incomplete overlap between networks. Second, exchangeable permutation priors are immediately applicable to other inferential problems in which permutations arise as latent parameters and the original labelings of the objects do not carry any information, such as record linkage and shuffled regression.

More broadly, while exchangeability provides a principled and tractable foundation, it is of interest to explore extensions beyond exchangeable permutations. Recent contributions show that exchangeable partitions can be obtained by marginalizing over suitably structured non-exchangeable permutations \citep{Gil–Leyva02012023}. These constructions could yield richer permutation prior classes while retaining interpretable marginal behavior at the level of cycle structures. Finally, one could think to  relax the assumption that the cycle structure of the permutation coincides exactly with the latent block structure of the networks, by introducing dependence between the two structures, for instance through conditionally partially exchangeable random partitions \citep{bea}. In such models, block structure would inform, but not deterministically constrain, the permutation cycles, allowing for cross-block matchings with low but nonzero probability. This added flexibility would increase modeling realism, at the cost of more challenging posterior exploration, motivating the development of new computational strategies for navigating complex permutation spaces.

Overall, this work opens the door to a broader class of Bayesian models for permutation-valued parameters, combining probabilistic rigor, interpretability, and algorithmic tractability, with applications well beyond the graph matching setting considered here.

\section*{Acknowledgments}
Lizhen Lin is funded by the NSF grant DMS 2503119. This research was largely conducted while Francesco Gaffi was a Postdoctoral Associate at the Department of Mathematics, University of Maryland, College Park.

\newpage
\begin{center}
\LARGE \bf Supplementary Material
\end{center}
\spacingset{1.8} 

\setcounter{section}{0}
\setcounter{equation}{0}
\setcounter{table}{0}
\setcounter{figure}{0}
\setcounter{algorithm}{0}

\renewcommand{\theequation}{S.\arabic{equation}}
\renewcommand{\thesection}{S\arabic{section}}  
\renewcommand{\thefigure}{S.\arabic{figure}}
\renewcommand{\thetable}{S.\arabic{table}}
\renewcommand{\thealgorithm}
{S.\arabic{algorithm}}

This supplementary material is structured as follows. In Section \ref{sec:ssp} we give a background on species sampling processes, we outline their connection to exchangeable random permutations and provide some special cases of exchangeable permutation probability functions that are included in our implementation. In Section \ref{sec:proofs}, proofs of theoretical results regarding exchangeable random permutations are given. In Section \ref{sec:conditional} we present an alternative model formulation and show the equivalence with the one given in Section \ref{sec:model}. In Section \ref{sec:fullco} we explicitate the derivations of the full conditionals employed in the MCMC algorithm presented in Section \ref{sec:post}. Finally, in Section \ref{sec:point_estim} we give a background on distances in the space of permutations, detail the posterior summary algorithm given in Section \ref{sec:summary} along with a fast approximated version of it.
\vspace{-18pt}
\section{Species sampling processes and EPerPFs}\label{sec:ssp}

As hinted in the main article, the strong connection between exchangeable partitions and species sampling processes (SSP) \citep[see \emph{e.g.}][]{pitman1995exchangeable,Pit(96),franzolini2025} carries over to exchangeable random permutations. A proper SSP is an almost-surely discrete random probability measure $\bP\overset{a.s.}{=}\sum_{h\geq1}\biw_{h}\delta_{\bchi_h}$ such that the atoms $(\bchi_h)_{h\geq1}$ are iid from a non-atomic probability measure $P_0$ on a Polish space, and independent of the random weights $\bw=(\biw_h)_{h\geq1}\sim\mathcal{L_\bw}$, for some distribution $\mathcal{L}_\bw$ on the infinite-dimensional simplex. If $\bP$ is a SSP, a species sampling model (SSM) is an exchangeable sequence $(\bix_n)_{n\geq1}$ directed by $\bP$, \emph{i.e.}, by \cite{de1937prevision},
\begin{equation}
    \label{eq:ssm}
    (\bix_n)_{n\geq1}\mid\bP\overset{iid}{\sim}\bP
\end{equation}
Now, being $\bP$ almost surely discrete, one has ties along the sequence $(\bix_n)_{n\geq1}$ with positive probability. Let $\bx_n^*=(\bix^*_1,\dots,\bix^*_{\bk_n})$ denote the unique values of $(\bix_1,\dots,\bix_n)$, in order of appearance. One can consider the sequence of random partitions $(\bz^{(n)})_{n\geq1}$ defined as
\begin{equation*}\label{eq:partition} 
\biz^{(n)}_{i}=j \quad \mbox{if and only if} \quad \bix_{ i}=\bix^*_{j}, \qquad \text{for }\, j\in[\bk_n],\,n\geq1
\end{equation*}
which is a consistent sequence of exchangeable random partitions, whose distribution is determined by an EPPF $(\varphi^{(n)})_{n\geq1}$. It is clear that the EPPF is completely identified by the law of the weights $\mathcal{L}_\bw$. In particular, by Theorem \ref{thm:char}, we can obtain an EPerPF as 
\begin{equation}\label{eq:moments}
p^{(n)}_k(n_1,\dots,n_k)=\frac{1}{\prod_{j=1}^k(n_j-1)!}\sum_{h_1\neq\dots\neq h_k}\mean{\prod_{j=1}^k\biw_{h_j}^{n_j}}
\end{equation}
for any $(n_1,\dots,n_k)$ with $\sum_{j=1}^kn_j=n$, any $k\in[n]$ and any $n\geq1$. Hence, when $\mathcal{L}_\bw$ is determined, one has a SSP, and in turn a permutation structure.

In the context of SSMs and random partition models, the most popular choice for $\mathcal{L}_\bw$ is the so-called Griffiths-Engen-McCloskey (GEM) distribution, which results in the original Chinese restaurant process (CRP) distribution for $(\bz^{(n)})$ and in $\bP$ having a Dirichlet process distribution \citep{Fer(73)}. The permutation structure deriving from this choice is detailed in Example \ref{ex:dirichlet}, where it is highlighted that under such distributions, the probability mass assigned to a permutation does not depend on the lengths of its cycles, but just on the number of cycles. Namely, in the position-aware Chinese restaurant metaphor, a new customer takes one of the seats already available at old tables \emph{equiprobably} across tables. This feature may be limiting in some scenarios, therefore it may be sensible to go beyond the Dirichlet process and explore other distributions for the random weights. To obtain them, one can work on the random probability measure $\bP$ as done in the literature on normalized random measures with independent increments \citep[NRMIs, see][]{Reg(03),Lijoi05,Lijoi07}, or address directly the construction of the random weights with stick-breaking techniques \citep{sethuraman94}. Here we list some possible choices for the distribution of $\bP$, which all result in Gibbs-type random partition distributions \citep{gnedinpitman,deblasi2015}, and give the corresponding exchangeable permutation probability function (EPerPF) and the predictive scheme of the related position-aware generalized Chinese restaurant process (PA-gCRP). 

\begin{exe}[Normalized stable process]\label{exe:nsp} \normalfont Introduced in \cite{Kin(75)}, the normalized stable process can be obtained by normalizing a homogeneous \emph{completely random measure} \citep[CRM, see][]{Kin(67),Reg(03)} whose jump distribution is given by a stable law. When $\bP$ in \eqref{eq:ssm} has such distribution, the corresponding EPerPF in \eqref{eq:moments} is
$$p_k^{(n)}(n_1,\dots,n_k)=\frac{\alpha^{k-1}(k-1)!}{(n-1)!}\prod_{j=1}^k\frac{\Gamma(n_j-\alpha)}{\Gamma(1-\alpha)(n_j-1)!}$$
for some $\alpha\in(0,1)$. The predictive scheme is
$$p(\pi_{n+1}\mid\pi_n)=\frac{1-\alpha/n_j}{n}\mathbbm{1}_{\{j\leq k\}}+\frac{k\alpha}{n}\mathbbm{1}_{\{j=k+1\}}$$
where $k=k(\pi_n)$ and $j=z_{n+1}(\pi_{n+1})$.
The predictive probability of creating a singleton grows linearly with the current number of cycles, but it is discounted by the parameter $\alpha$. Moreover, the seats at old tables are not equiprobable across-tables, as in the Dirichlet process case: those in large tables are more likely to be taken.
\end{exe}

\begin{exe}[Pitman-Yor process]\label{exe:pyp} \normalfont
Introduced in \cite{Pit(97)}, the two-parameter Poisson-Dirichlet distribution can be obtained normalizing the polynomial tilting of a stable CRM. The EPerPF of its position-aware version is
$$p_k^{(n)}(n_1,\dots,n_k)=\frac{\Gamma(\theta)}{\Gamma(\theta+n)}\prod_{j=1}^k\frac{(\theta+(i-1)\alpha)\Gamma(n_i-\alpha)}{\Gamma(1-\alpha)(n_i-1)!}$$
for some concentration parameter $\theta>0$, and some $\alpha\in(0,1)$. The predictive scheme is
$$p(\pi_{n+1}\mid\pi_n)=\frac{1-\alpha/n_j}{n+\theta}\mathbbm{1}_{\{j\leq k\}}+\frac{\theta+k\alpha}{n+\theta}\mathbbm{1}_{\{j=k+1\}}$$
where $k=k(\pi_n)$ and $j=z_{n+1}(\pi_{n+1})$.
Clearly, for $\theta=0$ one recovers the position-aware normalized stable of Example \ref{exe:nsp}, and for $\alpha=0$ one has the position-aware Dirichlet process of Example \ref{ex:dirichlet}.
\end{exe}

\begin{exe}[Gnedin process]\label{exe:gnp} \normalfont Introduced in \cite{gnedin2010}, the Gnedin process is obtained as a mixture of a Dirichlet multinomial distribution over the number of clusters $\bk$ through
$\prob{\bk=j}=\gamma\frac{\Gamma(j-\gamma)}{\Gamma(1-\gamma)}$ for some $\gamma\in(0,1)$. The EPerPF of the position aware version is
$$p_k^{(n)}(n_1,\dots,n_k)=\gamma(k-1)\frac{\Gamma(k-\gamma)\Gamma(n+\gamma-k)}{\Gamma(n+\gamma)\Gamma(n)}\prod_{j=1}^kn_j$$   The predictive scheme is
$$p(\pi_{n+1}\mid\pi_n)=\frac{n_j+1}{n_j}\frac{(n-k+\gamma)}{n(n+\gamma)}\mathbbm{1}_{\{j\leq k\}}+\frac{k(k-\gamma)}{n(n+\gamma)}\mathbbm{1}_{\{j=k+1\}}$$
where $k=k(\pi_n)$ and $j=z_{n+1}(\pi_{n+1})$. Among the presented exchangeable random permutations, the Gnedin is the only one featuring a finite (yet random) number of cycles when $n\rightarrow\infty$.
\end{exe}

\section{Proofs of Section \ref{sec:erp}}\label{sec:proofs}

\vspace{-15pt}

In this Section we give the proofs of the results stated in Section \ref{sec:erp} of the main article.
\begin{proof}[{\bf Proof of Proposition \ref{prp:cy_str}}]
We have that
\begin{align*}
    \prob{\zdr(\bpi)=(z_{\sigma(1)}(\pi),\dots,z_{\sigma(n)}(\pi))} &= \prob{\zdr(\bpi)=(z_1(\sigma\cdot\pi\cdot\sigma^{-1}),\dots,z_n(\sigma\cdot\pi\cdot\sigma^{-1})))} \\
    &= \prob{\zdr(\bpi)=(z_1(\pi),\dots,z_n(\pi))}
\end{align*}
    where the first equality holds because of Remark \ref{rmk:conjugate} and the second because of finite exchangeability of $\bpi$.
\end{proof}

\begin{proof}[{\bf Proof of Proposition \ref{prp:cons}}]
   For any $n\geq1$ we have that
\begin{align*}
    (z_1(\bpi_{n+1}),\dots,z_{n}(\bpi_{n+1}))&\overset{a.s.}{=} (z_1(\mathpzc{d}(\bpi_{n+1})),\dots,z_{n}(\mathpzc{d}(\bpi_{n+1}))) \\
    &\overset{d}{=} (z_1(\bpi_{n}),\dots,z_{n}(\bpi_{n}))
\end{align*}
The first equality holds because the projection $\mathpzc{d}(\cdot)$ leaves the cycle structure and the cycle labeling (thanks to the order of appearance) of the first $n$ elements unchanged. The second equality is given by consistency of the sequence $(\bpi_m)_{m\geq1}$ with respect to the projection $\mathpzc{d}(\cdot)$.
\end{proof}

\begin{proof}[{\bf Proof of Theorem \ref{thm:char}}]
In general, the probability mass function (pmf) of a random permutation factorizes as
\begin{equation*}
    p(\pi)=p\left(\pi\mid\zdr(\pi)\right)p\left(\zdr(\pi)\right)
\end{equation*}
If $\bpi$ is exchangeable, then there is a consistent sequence of finitely exchangeable random permutations $(\bpi^\star_m)_{m\geq1}$ such that $\bpi^\star_n\overset{d}{=}\bpi$. By Propositions~\ref{prp:cy_str} and \ref{prp:cons}, $(\zdr(\bpi^\star_m))_{m\geq1}$ is a consistent sequence of finitely exchangeable partitions, and clearly $\zdr(\bpi^\star_n)\overset{d}{=}\zdr(\bpi)$. Hence the marginal distribution of $\zdr(\bpi)$ must be given by an EPPF. Moreover, if two permutations have the same cycle structure, they also have the same cycle type.
Therefore, by Definition~\ref{def:rand_perm}, the conditional distribution of $\bpi$ given $\zdr(\bpi)=\zdr$ must be the uniform distribution over all the permutations with $\zdr$ as cycle structure, resulting in the expression in \eqref{eq:pmf}.

Conversely, \eqref{eq:pmf} clearly entails equal probability on permutations with the same cycle type. Hence if $\bpi$ is distributed as \eqref{eq:pmf}, then it is finitely exchangeable. Now, let us consider the sequence $(\varphi^{(m)})_{m\geq1}$. Because of consistency conditions of the kind expressed in \eqref{eq:proj_eppf}, by Kolmogorov existence theorem, there exists a sequence of random partitions $(\bz^\star_m)_{m\geq1}$ whose finite dimensional distributions are determined by $(\varphi^{(m)}_k)_{k\leq m}$ for any $m\geq1$. Define $(\bpi^\star_m)_{m\geq1}$ as the sequence of random permutations that are uniformly distributed conditionally on $(\bz^\star_m)_{m\geq1}$. The law of such sequence is determined by $(p^{(m)})_{m\geq1}$ obtained from $(\varphi^{(m)})_{m\geq1}$ as in \eqref{eq:pmf}, which means that every $\bpi^\star_m$ is finitely exchangeable. Now let $\brho=\mathpzc{d}(\bpi^\star_{n+1})$, we want to show that $\brho\overset{d}{=}\bpi^\star_n$. Let $\rho\in\mathcal{S}_{n}$ with $k$ cycles and cycle lengths given by $\cdr(\rho)=(n_1,\dots,n_k)$. We have
\begin{equation}\label{eq:proj1}
    \prob{\brho=\rho}=\sum_{\sigma\in\mathcal{A}(\rho)}\prob{\bpi^\star_{n+1}=\sigma}
\end{equation}
It is easy to see that in $\mathcal{A}(\rho)$ there are exactly $n_j$ permutations $\sigma$ such that $\cdr(\sigma)~=(n_1,\dots,n_j~+~1,\dots,n_k)$ for any $j\in[k]$, and only one such that $\cdr(\sigma)=(n_1,\dots,n_k,1)$.
Hence \eqref{eq:proj1} equals
\begin{equation*}\label{eq:proj2}
\begin{split}
    \sum_{j=1}^k\frac{n_j}{\prod_{\substack{i=1\\i\neq j}}^k(n_i-1)!n_j!}\varphi^{(n+1)}_k(n_1,\dots,n_j+1,\dots,n_k)
    +~ \frac{1}{\prod_{i=1}^k(n_i-1)!}\varphi^{(n+1)}_{k+1}(n_1,\dots,n_k,1) \\
    =~ \frac{1}{\prod_{i=1}^k(n_i-1)!}\varphi^{(n)}_{k}(n_1,\dots,n_k)=p_k^{(n)}(n_1,\dots,n_k)
    \end{split}
\end{equation*}
by \eqref{eq:proj_eppf}. Since clearly $\bpi^\star_n\overset{d}{=}\bpi$, we proved that if $\bpi$ has pmf given in \eqref{eq:pmf}, then it is exchangeable by Definition \ref{def:experm}.
\end{proof}

\begin{proof}[{\bf Proof of Theorem \ref{thm:eperpf}}]
    The necessity of \eqref{eq:cons} given a sequence of the form in \eqref{eq:pmf} is showed in Proof of Theorem \ref{thm:char}. For sufficiency, let us prove by induction on $n$ that $p^{(n)}$ satisfying \eqref{eq:cons} is a pmf on $\mathcal{S}_n$. Clearly, it is true for $p^{(1)}$. If we assume it is for $p^{(n)}$ we have
    \begin{equation}\label{eq:sumto1}\sum_{\pi\in\mathcal{S}_n}\sum_{\sigma\in\mathcal{A}(\pi)}p_{k(\sigma)}^{(n+1)}(\cdr(\sigma))=1\end{equation}
    It is easy to see that $\mathcal{S}_{n+1}=\bigcup_{\pi\in\mathcal{S}_n}\mathcal{A}(\pi)$ and $\mathcal{A}(\pi)\cap\mathcal{A}(\pi')=\varnothing$ for $\pi\neq\pi'$, hence \eqref{eq:sumto1} gives $\sum_{\sigma\in\mathcal{S}_{n+1}}p_{k(\sigma)}^{(n+1)}(\cdr(\sigma))=1$. Now, we prove that a random partition in $\mathcal{S}_n$ with pmf defined in \eqref{eq:cons} is finitely exchangeable. For any $\pi,\sigma\in\mathcal{S}_n$, if $k(\pi)=k$, we have
    \begin{equation*}
        p^{(n)}_{k}(\cdr(\sigma\cdot\pi\cdot\sigma^{-1}))=p^{(n)}_{k}(c_{\rho(1)}(\pi),\dots,c_{\rho(k)}(\pi))=p^{(n)}_{k}(\cdr(\pi))
    \end{equation*}
    where the first equality holds because $\tdr(\pi)=\tdr(\sigma\cdot\pi\cdot\sigma^{-1})$ and Remark \ref{rmk:permind}, while the second holds because $p^{(n)}$ is symmetric. By Remark \ref{rmk:finex}, we know that finite exchangeability implies the following decomposition of the pmf
    \begin{equation}\label{eq:decom}
    p^{(n)}_{k(\pi)}(\cdr(\pi))=\frac{1}{\prod_{j=1}^{k(\pi)}(c_j(\pi)-1)!}f^{(n)}(\cdr(\pi))
    \end{equation}
    where $f^{(n)}$ is a symmetric function of $\cdr(\pi)$ that gives the probability that the cycle structure of the random permutation has those lengths. Now, substituting \eqref{eq:decom} in \eqref{eq:cons}, and using, as in Proof of Theorem \ref{thm:char}, that in $\mathcal{A}(\pi)$ there are exactly $c_j(\pi)$ permutations $\sigma$ such that $c_i(\sigma)=c_i(\pi)$ for any $i\neq j$ and $c_j(\sigma)=c(\pi)+1$ for any $j=1,\dots,k(\pi)$, and one permutation such that $k(\sigma)=k(\pi)+1$ and $c_{k(\sigma)}(\sigma)=1$, we have
    \begin{align*}
&f^{(n)}(c_1(\pi),\dots,c_{k(\pi)}(\pi))=\\
    \sum_{j=1}^k&f^{(n+1)}(c_1(\pi),\dots,c_j(\pi)+1,\dots,c_{k(\pi)}(\pi))
    +~f^{(n+1)}(c_1(\pi),\dots,c_{k(\pi)},1) 
\end{align*}
which, together with symmetry and the obvious condition $f^{(1)}(1)=1$, implies that $(f^{(n)})_{n\geq1}$ is an EPPF \citep[see][]{Pit(96)}.
\end{proof}

\begin{proof}[{\bf Proof of Proposition \ref{prp:pred}}]
  By Theorem~\ref{thm:char}, the pmf of $\bpi_n$ is given in \eqref{eq:pmf} for any $n\geq1$, while we have
    \begin{equation}\label{eq:predproof1}
        \cprob{\bpi_{n+1}=\sigma}{\bpi_n=\pi} =\frac{\prob{\bpi_{n+1}=\sigma,\bpi_n=\pi}}{\prob{\bpi_n=\pi}}=\frac{\prob{\bpi_{n+1}=\sigma}}{\prob{\bpi_n=\pi}}\mathbbm{1}_{\mathcal{A}(\pi)}(\sigma)
    \end{equation}
    Hence, substituting \eqref{eq:pmf} in \eqref{eq:predproof1} one easily obtains \eqref{eq:pred} for $\sigma\in\mathcal{A}(\pi)$ and $0$ otherwise.
   
    On the other hand, let us determine the joint distribution of a sequence $(\bpi_1,\dots,\bpi_n)$ with $\bpi_i\in\mathcal S_i$ for $i=1,\dots,n$ whose predictive structure is given by \eqref{eq:pred}. It is easy to see that if $\pi_i\neq\mathpzc{d}(\pi_{i+1})$ for any $i=1,\dots,n-1$ then $p(\pi_1,\dots,\pi_n)=0$. Consider instead a sequence $(\pi_1,\dots,\pi_n)$ such that $\pi_i=\mathpzc{d}(\pi_{i+1})$ for any $i=1,\dots,n-1$. Then $p(\pi_n)=p(\pi_1,\dots,\pi_n)$ and we have \begin{equation}
\label{eq:last}
\begin{split}p(\pi_1,\dots,\pi_n)=\cprob{\bpi_n=\pi_n}{\bpi_{n-1}=\pi_{n-1}}\cdots\prob{\bpi_1=\pi_1}=\\\varphi^{(n)}_{k(\pi_n)}(n^{(n)}_1,\dots,n_{k(\pi_n)}^{(n)})\prod_{i=1}^n\frac{1}{(n^{(i)}_{z_i(\pi_i)}-1)\wedge 1}\end{split}\end{equation}
    where $n^{(i)}_j$ is the number of nodes in the $j$-th cycle of $\pi_i$ in order of appearance. There is a one-to-one correspondence between the factors in \eqref{eq:last} and the ones in $\prod_{j=1}^{k(\pi_n)}\frac{1}{(n_j^{(n)}-1)!}$ if one expands the factorials. Hence we recovered the pmf in \eqref{eq:pmf}.
\end{proof}
\vspace{-18pt}

\section{Alternative model formulation}\label{sec:conditional}

It is possible to represent the model given in \eqref{eq:model1} without employing a common parent network.
In this case, $\bY^{(1)}$ is an SBM with a block probability matrix whose entries are constrained in the interval $(\alpha,1-\beta)$, while $\bY^{(2)}$, modulo the permutation, is a noisy observation of $\bY^{(1)}$ given some specific error rates. In fact, marginalizing out the parent network $\bY$ from the model in \eqref{eq:model1}, the entries of $\bY^{(1)}$ and $\bY^{(2)}$ are independent across rows and columns and we have
\begin{equation}\label{eq:jointprob}
\begin{split}
\mathds{P}_{\bXi,\bpi}\left[\biy_{uv}^{(1)}=y,\,\biy_{\bpi(u)\bpi(v)}^{(2)}=y'\right]=\\=
\mathds{P}_{\bXi,\bpi}\left[\biy_{uv}^{(1)}=y,\,\biy_{\bpi(u)\bpi(v)}^{(2)}=y'\mid\biy_{uv}=1\right]\mathds{P}_{\bXi,\bpi}\left[\biy_{uv}=1\right]+\\+\mathds{P}_{\bXi,\bpi}\left[\biy_{uv}^{(1)}=y,\,\biy_{\bpi(u)\bpi(v)}^{(2)}=y'\mid\biy_{uv}=0\right]\mathds{P}_{\bXi,\bpi}\left[\biy_{uv}=0\right]=\\
=(1-\beta)^{y+y'}\beta^{(1-y)+(1-y')}\bxi_{z_u(\bpi)z_v(\bpi)}+\alpha^{y+y'}(1-\alpha)^{(1-y)+(1-y')}(1-\bxi_{z_u(\bpi)z_v(\bpi)})
\end{split}
\end{equation}
Then, defining $\bPsi:=(1-\beta)\,\bXi+\alpha\,(1-\bXi)$ and $\bLambda:=\frac{(1-\beta)^2\,\bXi+\alpha^2\,(1-\bXi)}{\bPsi}$, it is easy to see that
model \eqref{eq:model1} is equivalent to

\begin{equation}\label{eq:model2}
\begin{aligned}
        \biy^{(2)}_{\bpi(u)\bpi(v)}\mid\biy^{(1)}_{uv},\bPsi,\bLambda,\bpi
        \sim&\begin{cases}
        \text{Bern}(\blambda_{z_u(\bpi)z_v(\bpi)})&\text{if }\,\biy^{(1)}_{uv}=1\\
        \text{Bern}\left(\frac{\bpsi_{z_u(\bpi)z_v(\bpi)}(1-\blambda_{z_u(\bpi)z_v(\bpi)})}{1-\bpsi_{z_u(\bpi)z_v(\bpi)}}\right)&\text{if }\,\biy^{(1)}_{uv}=0
    \end{cases}\\
    \biy^{(1)}_{uv}\mid\bPsi,\bpi\overset{ind}{\sim}&\,\text{Bern}(\bpsi_{z_u(\bpi)z_v(\bpi)})\\
        \bXi&\overset{iid}{\sim}\text{beta}(a_\xi,b_\xi)\\
    \bpi&\sim \text{PA-gCRP}(\varphi^{(n)})
    \end{aligned}
    \end{equation}  
Notice that this derivation is only possible because $z_u(\bpi)=z_{\bpi(u)}(\bpi)$ for any $u\in[n]$, that is two matched nodes are always in the same latent block. The representation in \eqref{eq:model2} carries over from the correlated Erd\H{o}s-R\'enyi case \citep{nate}.
Using again \eqref{eq:jointprob}, this version of the model is characterized by the following joint distribution
\begin{equation}\label{eq:joint2}
\begin{split}
p(\Ydr^{(1)},\Ydr^{(2)},\Xi,\pi)
    =\prod_{j\leq h}^{k(\pi)}\left[(1-\beta)^2\xi_{jh}+\alpha^2(1-\xi_{jh})\right]^{m_{jh}^1}\left[\beta(1-\beta)\xi_{jh}+\alpha(1-\alpha)(1-\xi_{jh})\right]^{m_{jh}^\times}\times\\
    \times \left[\beta^2\xi_{jh}+(1-\alpha)^2(1-\xi_{jh})\right]^{m_{jh}^0}\frac{\xi_{jh}^{a-1}(1-\xi_{jh})^{b-1}}{\mbox{B}(a,b)}p_k^{(n)}(n_1,\dots,n_k)
\end{split}
\end{equation}
where $m_{jh}^1=\sum_{u,v=1}^ny^{(1)}_{uv}y^{(2)}_{\zdr_u(\bpi)\zdr_v(\bpi)}\mathbbm{1}_{\{\zdr_u(\bpi)=j,\zdr_v(\bpi)=h\}}$ is the number of \emph{concordant edges} in $(\Ydr^{(1)},\Ydr_\pi^{(2)})$ between cycle $j$ and cycle $h$; similarly $m^\times_{jh}$ is the number of \emph{discordant edges}, while $m^0_{jh}$ is the number of \emph{concordant non-edges}. 
\section{Derivation of full conditionals}\label{sec:fullco}
Given model \eqref{eq:model1}, we have
\begin{equation*}
\begin{split}
     p(\Ydr^{(2)},\Ydr^{(1)},\Ydr,\pi\mid\alpha,\beta)= \beta^{e_1^{(1)}(\mathfrak{id})+e_1^{(2)}(\pi)}(1-\beta)^{\overline{e}_1^{(1)}(\mathfrak{id})+\overline e_1^{(2)}}\alpha^{e_0^{(1)}(\mathfrak{id})+e_0^{(2)}(\pi)}(1-\alpha)^{\overline e_0^{(1)}(\mathfrak{id})+\overline e_0^{(2)}(\pi)}\\
    \prod_{j\leq h}^{k(\pi)}\frac{\mbox{B}(1;a_\xi+m_{jh},b_\xi+\overline{m}_{jh})}{\mbox{B}(1;a_\xi,b_\xi)}
    p_{k(\pi)}^{(n)}(\cdr(\pi))
\end{split}
\end{equation*}
For step 1(i), we have
\begin{equation*}
\begin{split}
    p_{\balpha,\bbeta}(\pi^\star \mid \Ydr, \Ydr^{(1)}, \Ydr^{(2)}, \mathpzc{d}_v(\pi))=
    \frac{p_{\balpha,\bbeta}(\Ydr^{(2)},\Ydr^{(1)},\Ydr\mid\pi^\star)p(\pi^\star\mid\mathpzc{d}_v(\pi))}{p_{\balpha,\bbeta}(\Ydr^{(2)},\Ydr^{(1)},\Ydr\mid\mathpzc{d}_v(\pi))}=\\
    \frac{p_{\balpha,\bbeta}(\Ydr^{(2)}\mid\Ydr,\pi^\star)p_{\balpha,\bbeta}(\Ydr^{(1)}\mid\Ydr)p(\Ydr\mid\pi^\star)p(\pi^\star\mid\mathpzc{d}_v(\pi))}{p_{\balpha,\bbeta}(\Ydr^{(2)}\mid\Ydr,\mathpzc{d}_v(\pi))p_{\balpha,\bbeta}(\Ydr^{(1)}\mid\Ydr)p(\Ydr\mid\mathpzc{d}_v(\pi))}
    \end{split}
\end{equation*}
Now, $\bY$ only depends on the permutation $\bpi^\star$ through its cycle structure $\zdr(\bpi^\star)$. Moreover $\frac{p_{\balpha,\bbeta}(\Ydr^{(2)}\mid\Ydr,\pi^\star)}{p_{\balpha,\bbeta}(\Ydr^{(2)}\mid\Ydr,\mathpzc{d}_v(\pi))}\propto \frac{p_{\balpha,\bbeta}(\Ydr^{(2)}\mid\Ydr,\pi^\star)}{p_{\balpha,\bbeta}(\Ydr^{(2)}_{-v}\mid\Ydr_{-v},\mathpzc{d}_v(\pi))}$ and 
$\frac{p_{\balpha,\bbeta}(\Ydr\mid\Ydr,\pi^\star)}{p_{\balpha,\bbeta}(\Ydr^{(2)}\mid\Ydr,\mathpzc{d}_v(\pi))}\propto \frac{p_{\balpha,\bbeta}(\Ydr^{(2)}\mid\Ydr,\pi^\star)}{p_{\balpha,\bbeta}(\Ydr^{(2)}_{-v}\mid\Ydr_{-v},\mathpzc{d}_v(\pi))}$ by $\pi^\star$-independent constants. Hence we get the full conditional in \eqref{eq:mh}. The derivation of usable expressions for the three factors is given in the main article.

\noindent For step 1(ii), we have
\begin{equation}\label{eq:parent}
\begin{split}   p_{\bpi,\balpha,\bbeta}(y_{vu}\mid y_{v1},\dots,y_{vu-1},\Ydr_{-v},\Ydr^{(1)},\Ydr^{(2)})=\\\frac{p_{\bpi,\balpha,\bbeta}(\Ydr^{(1)},\Ydr^{(2)}\mid y_{v1},\dots,y_{vu},\Ydr_{-v})p_{\zdr(\bpi)}(y_{v1},\dots,y_{vu},\Ydr_{-v})}{p_{\bpi,\balpha,\bbeta}(y_{v1},\dots,y_{vu-1},\Ydr_{-v},\Ydr^{(1)},\Ydr^{(2)})}
\end{split}
\end{equation}
where
\begin{equation*}
 \begin{split}p_{\bpi,\balpha,\bbeta}(\Ydr^{(1)},\Ydr^{(2)}\mid y_{v1},\dots,y_{vu},\Ydr_{-v})=\\
 \prod_{\ell=1,2}p_{\bpi,\balpha,\bbeta}(y_{vu}^{(\ell)}\mid y_{vu})\left\{\prod_{u'=1}^{u-1}p_{\bpi,\balpha,\bbeta}(y^{(\ell)}_{vu'}\mid y_{vu'})\right\}p_{\bpi,\balpha,\bbeta}(\Ydr^{(\ell)}_{-v}\mid \Ydr_{-v})\times \\p_{\bpi,\balpha,\bbeta}(y^{(1)}_{vu+1},\dots,y^{(1)}_{vn},y^{(2)}_{vu+1},\dots,y^{(2)}_{vn}\mid y_{v1},\dots,y_{vu},\Ydr_{-v})
 \end{split}
\end{equation*}
and
\begin{equation*}
  p_{\zdr(\bpi)}(y_{v1},\dots,y_{vu},\Ydr_{-v})=\prod_{j\leq h}^{k(\bpi)}\frac{\mbox{B}(a_\xi+m^{-vu+1}_{jh},b_\xi+\overline{m}_{jh}^{-vu+1})}{\mbox{B}(a_\xi,b_\xi)} 
\end{equation*}
with $m^{-vu+1}_{jh},\,\overline{m}_{jh}^{-vu+1}$ giving the edges and non-edges count in $\Ydr$ between cycles $j$ and $h$ disregarding the interactions between node $v$ and nodes $u'\geq u+1$. Hence, considering $y_{vu}$-dependent factors only, \eqref{eq:parent} is proportional to \eqref{eq:prob1} for $y_{vu}=1$ and to \eqref{eq:prob0} for $y_{vu}=0$.

\indent For step 2, the full conditionals of the noise rates $\balpha,\bbeta$ are easily obtained by beta-binomial conjugacy.
\vspace{-18pt}
\section{More on point estimation}\label{sec:point_estim}

\subsection{On distances in $\mathcal S_n$}\label{sec:dist}
 Several distances between permutations have been considered in statistical applications \citep[see \emph{e.g.} Chapter 6 in][]{Dia(88)}. The \emph{Hamming distance}, which we denote as $d_H(\cdot,\cdot)$, simply counts the output mismatches between the two maps, \emph{i.e.} $d_H(\pi,\sigma)~=~\sum_{i=1}^n\mathbbm{1}_{\{\pi(i)=\sigma(i)\}}$. The \emph{Cayley distance}, denoted as $d_C(\cdot,\cdot)$, is the minimum number of transpositions needed to transform a permutation into another, and can be easily evaluated via Cayley's identity. Both these distances are \emph{bi-invariant}, that is left invariant, \emph{i.e.} $d(\rho\cdot\pi,\rho\cdot\sigma)=d(\pi,\sigma)$ for any $\rho\in\mathcal{S}_n$, and right invariant, \emph{i.e.} $d(\pi,\sigma)=d(\pi\cdot\rho',\sigma\cdot\rho')$ for any $\rho'\in\mathcal{S}_n$. This property translates into invariance to the initial and final labeling of the objects, which makes these distances suitable for contexts where the permutation represents a matching between two sets of objects that are not initially ordered, as in the case of graph matching. This is not the case, \emph{e.g.}, when one deals with random ranks, as in Bradley-Terry-Luce or Mallows models \citep{bradleyterry,luce,mallows}. In those models the output labels of the permutations (which are observed) are ordered and metrized. For example, if one observes $\pi(3)=1>\pi(4)=2>\pi(2)=3>\pi(1)=4$, this means: 3 is ranked best than 4, which is ranked best than 2, which is ranked best than 1. Therefore, distances that can disentangle an ordering are employed in this case, including the \emph{Spearman's footrule} $d_1(\pi,\sigma)=\sum_{i=1}^n|\pi(i)-\sigma(i)|$, the \emph{Spearman's rank correlation} $d_2(\pi,\sigma)=\sum_{i=1}^n(\pi(i)-\sigma(i))^2$, and the \emph{Kendall-$\tau$ distance}, defined as the minimum number of pairwise adjacent transpositions to transform a permutation into another. These are all right invariant (\emph{i.e.} do not consider the original ordering of the inputs) but not left invariant, since they depend on the output labels. If the minimization in \eqref{eq:postexpdist} is done with respect to the Kendall-$\tau$ distance, its solution is the Kemeny consensus ranking \citep{kemeny} of the posterior sample $(\pi_s)_{s=1}^S$.

Among the presented bi-invariant distances, we choose to search for a point estimate $\hat\pi$ that minimizes Cayley distance from the sample. In fact, Hamming distance  may drastically disregard the cycle structure information. For example, if $\pi=(123)(456)$ and $\pi'=(132)(465)$ then $d_H(\pi,\pi')=6$, even if $\zdr(\pi)=\zdr(\pi')$, whereas $d_C(\pi,\pi')=4$. In Cayley distance, permutations with nested cycle structures are generally close: if $\sigma~=~(13)(2)(46)(5)$ then $d_C(\pi,\sigma)=2$ while $d_H(\pi,\sigma)=4$. However, Cayley distance does not metrize $\mathcal S_n$ so that similarity in cycle structure always implies lower distance. It is easy to prove that if $\zdr(\pi)=\zdr(\pi')$ and $\pi\neq\pi'$ then $d_C(\pi,\pi')>1$, while there exist $\sigma'$ such that $\zdr(\sigma')\neq\zdr(\pi)$ and $d_C(\pi,\sigma')=1$. For example, take again $\pi=(123)(456),\,\pi'=(132)(465)$ and $\sigma'=(12)(3)(456)$.

In conclusion, even if the choice of distance does not completely align with the specifics of our posterior sample, we manage to obtain a point estimate that successfully targets the true matching in the different simulation scenarios showed in Section \ref{sec:sim}.
\begin{algorithm}[t!]
	\caption{perSALSO}\label{alg:2}
	\begin{algorithmic}
\footnotesize
\State {\bf Input:} posterior sample $(\pi_s)_{s=1}^S$, $n_{zeal}$ maximum zealous updates

    \State 1. set $\brho\leftarrow\texttt{sample}(1 \ldots, n)$ \Comment Initialization phase
    \State 2. set $\hat\pi^{(1)}_1\leftarrow (\brho(n))$
	\For{$i=2,\dots,n$}
		\State prune $(\pi_s)_{s=1}^S$ from nodes $\brho(1),\dots,\brho(n-i)$: compute $(\mathpzc{d}_{\mathcal V^{(\brho)}_{n-i}}(\pi_s))_{s=1}^S$
        \State order $\mathcal A_{\brho(n-i+1)}(\hat\pi^{(1)}_{i-1})$: $\sigma$ s.t. $\brho\cdot\sigma(n-i+1)$ is most frequent in $\left(\mathpzc{d}_{\mathcal V^{(\brho)}_{n-i}}\left(\brho\cdot\pi_s(n-i+1)\right)\right)_{s=1}^S$ first
    \State set $min\leftarrow\texttt{Inf}$
    \For{$\sigma\in\mbox{ordered }\mathcal A_{\brho(n-i+1)}(\hat\pi^{(1)}_{i-1})$}
        \State compute $\tilde f^{(i)}_C(\sigma):=\frac{1}{S}\sum_{s=1}^Sd_C(\mathpzc{d}_{\mathcal{V}^{(\brho)}_{n-i}}(\pi_s),\sigma)$ with early stopping
        \State{\bf if\;} {$\tilde f^{(i)}_C(\sigma)<min$} {\bf \;then\;} 
        set $\hat\pi^{(2)}_i\leftarrow \sigma$; set $min\leftarrow \tilde f_C^{(i)}(\sigma)$ 
    \EndFor
    \EndFor
    \State 3. set $\hat\pi^{(2)}_0\leftarrow \hat\pi^{(1)}_n$ \Comment Sweetening phase
        \State 4. set $\brho'\leftarrow \texttt{sample}(1,\dots,n)$
        \For{$i =1,\dots,n$}
        \State order $\mathcal{A}_{\brho'(i)}(\mathpzc{d}_{\brho'(i)}(\hat\pi_{i-1}^{(2)}))$: $\sigma$ s.t. $\brho'\cdot\sigma(i)$ is most frequent in $\left(\brho'\cdot\pi_s(i)\right)_{s=1}^S$ first
		\For{$\sigma\in\mbox{ordered }\mathcal{A}_{\brho'(i)}(\mathpzc{d}_{\brho'(i)}(\hat\pi_{i-1}^{(2)}))$}
        \State compute $ f_C(\sigma)$ with early stopping
        \State{\bf if\;}{$f_C(\sigma)<min$}
        {\bf \;then\;} set $\hat\pi^{(1)}_i\leftarrow \sigma$; set $min\leftarrow f_C(\sigma)$
    \EndFor
    \EndFor
    \State{\bf if\;} { $f_C(\hat\pi^{(2)}_n)<f_C(\hat\pi^{(2)}_0)$ } {\bf \;then\;} set $\hat\pi^{(2)}_0\leftarrow \hat\pi^{(2)}_n$; {\bf go to } 4.
    \State 5. set $\hat\pi^{(3)}_0\leftarrow\hat\pi^{(2)}_n$
    \State 6. set $\kappa\leftarrow k(\hat\pi^{(3)}_0)$ \Comment{Zealous updates phase} 
\State{\bf repeat} $n_{zeal}$ times
\State\hspace{\algorithmicindent} set $\bij \leftarrow \texttt{unif}(1,\dots,\kappa)$
\State\hspace{\algorithmicindent} set $\brho_\bij\leftarrow\texttt{sample}(1,\dots,n_\bij)$ ($n_\bij$: $\bij$-th cycle length of $\hat \pi^{(3)}_0$)
\State\hspace{\algorithmicindent} set $\hat\pi^{(3)}_1\leftarrow \mathpzc{d}_{\mathcal C_{\bij}\setminus\{\brho_\bij(n_\bij)\}}(\hat \pi^{(3)}_0)$ ($\mathcal C_{\bij}$: $\bij$-th cycle of $\hat \pi^{(3)}_0$)
\State\hspace{\algorithmicindent} {\bf for} {$i=2,\dots,n_\bij$} {\bf do}
\State\hspace{\algorithmicindent}\hspace{\algorithmicindent} prune $(\pi_s)_{s=1}^S$ from nodes $\brho_\bij(1),\dots,\brho_\bij(n_\bij-i)$: compute $(\mathpzc{d}_{\mathcal V^{(\brho_\bij)}_{n-i}}(\pi_s))_{s=1}^S$\State\hspace{\algorithmicindent}\hspace{\algorithmicindent} order $\mathcal A_{\brho_\bij(n_\bij-i+1)}(\hat\pi^{(3)}_{i-1})$ as in Initialization
\State\hspace{\algorithmicindent}\hspace{\algorithmicindent} reset $min\leftarrow\texttt{Inf}$
\State\hspace{\algorithmicindent}\hspace{\algorithmicindent} {\bf for} {$\sigma\in\mbox{ordered }\mathcal A_{\brho_\bij(n_\bij-i+1)}(\hat\pi^{(3)}_{i-1})$} {\bf do}
    \State \hspace{\algorithmicindent}\hspace{\algorithmicindent}\hspace{\algorithmicindent} compute $\tilde f^{(i)}_C(\sigma):=\frac{1}{S}\sum_{s=1}^Sd_C(\mathpzc{d}_{\mathcal{V}^{(\brho_\bij)}_{n_\bij-i}}(\pi_s),\sigma)$ with early stopping
\State\hspace{\algorithmicindent}\hspace{\algorithmicindent}\hspace{\algorithmicindent}{\bf if\; } {$\tilde f^{(i)}_C(\sigma)<min$}
        {\bf \;then\;} set $\hat\pi^{(3)}_i\leftarrow \sigma$; set $min\leftarrow \tilde f_C^{(i)}(\sigma)$
\State\hspace{\algorithmicindent}\hspace{\algorithmicindent}{\bf end for}
\State\hspace{\algorithmicindent}\hspace{\algorithmicindent}{\bf if\;} {$f_C(\hat\pi^{(3)}_{n_\bij})<f_C(\hat\pi^{(3)}_{0})$} {\bf\;then\;} {set $\hat\pi^{(3)}_0\leftarrow\hat\pi^{(3)}_{n_\bij}$}
\State {\bf end repeat}
\State set $\hat\pi\leftarrow\hat\pi^{(3)}_0$
\State{\bf Output}: $\hat\pi$ 
	\end{algorithmic}
\end{algorithm}
\subsection{Speed-ups for perSALSO}\label{sec:speed}

The posterior summary procedure described in Section \ref{sec:summary} is detailed in Algorithm \ref{alg:2}. 

In the following we describe some techniques we implemented for both speeding up the minimization procedure and cutting the cost of each evaluation of $f_C(\cdot)$.
The discreteness of $\mathcal S_n$ often causes ties in the values of $f_C(\cdot)$.  To avoid the greedy search to be stuck jumping back and forth between equivalent local minima, in all the phases, we accept a move to $\sigma$ just if it is \emph{strictly} better than $\hat\pi_{i-1}$. On the other hand, this may cause poor exploration of the space. To counterbalance, we order the set of available new states so that we check first the ones matching the new node to outputs which, in the sample, are more frequent for that node, as described in Algorithm \ref{alg:2}. The evaluation of $f_C$ (or $\tilde f_C$, which in Algorithm \ref{alg:2} is the average Cayley distance from the pruned sample) requires in principle to scan the full sample $(\pi_s)_{s=1}^S$. We can speed up the minimization with an easy early stopping scheme that cuts short the computation of $f_C(\sigma)$ and discards $\sigma$ as soon as $\sum_{s=1}^{\tilde s}d_C(\pi_s,\sigma)>Sf_C(\hat\pi_{i-1})$ for some $\tilde s\leq S$; same for the pruned sample in phases (1) and (3). Furthermore, even when employing Cayley's identity, the calculation of each $d_C(\pi_s,\sigma)$ may still be computationally expensive, since it  involves the inversion of $\sigma$, which requires a full scan of the permutation vector. We can speed this up by devising the function that obtains all the candidates $\sigma\in\mathcal A_{\brho(n-i+1)}(\hat\pi_{i-1})$ to output also their inverses, constructing them from the inverse of the current state $\hat\pi_{i-1}$ with the same logic as in \eqref{eq:defstar}.

Still, given a constant cost $C$ for the computation of Cayley distance between two permutations, each pass of phase (2) of perSALSO (which is the most expensive) is $\mathcal O(CSn^2)$ (not considering the early stopping). We propose a two-step procedure that cuts the cost of the minimization by constraining each $\mathcal A_{\brho(n-i+1)}(\hat\pi_{i-1})$. Specifically, we first find a point estimate for the cycle structure of $\hat \zdr$ by employing the original SALSO \citep{salso} on the sample $(\zdr(\pi_s))_{s=1}^S$. Then we run our perSALSO with the constraint that all the searches are limited to permutations having $\hat \zdr$ as cycle structure. This means that each node reallocation in phase (2) requires $n_j$ evaluations of $f_C$ instead of $n$. This fast version is not guaranteed to give the same point estimate as full perSALSO, even seeding all the random operations, since, in general, $\zdr(\hat\pi)\neq\hat\zdr$. However, in our applications, this is mostly the case, which is not surprising given the importance that the cycle structure information has in our model.

\vspace{10pt}
\let\oldbibliography\thebibliography
\renewcommand{\thebibliography}[1]{\oldbibliography{#1}
	\setlength{\itemsep}{4.5pt}} 

\spacingset{1}
\begingroup
\fontsize{11pt}{11pt}\selectfont
\bibliographystyle{jasa3}
\bibliography{biblio.bib}
\endgroup

\end{document}